\documentclass[12pt]{iopart}

\usepackage{tikz}
\usetikzlibrary{arrows.meta,intersections}
\usepackage{tkz-euclide}
\usetkzobj{all}
\tikzset{>={Stealth[width=3mm,length=3mm]}}

\usepackage[utf8x]{inputenc}
\usepackage{amssymb}
\usepackage{amsthm}
\usepackage{amsopn}
\usepackage{doi}

\usepackage[strings]{underscore} 

\usepackage{graphicx}
\usepackage{xcolor}
\usepackage{amsbsy}
\hypersetup{colorlinks=true, allcolors=blue}
\usepackage{xifthen}

\usepackage{microtype}

\DeclareGraphicsExtensions{.pdf,.png,.jpg,.eps}

\newif\ifdraft

\draftfalse

\ifdraft
    \newcommand{\ja}[1]{\textcolor{purple}{JoakimA: #1}}
    \newcommand{\jatoah}[1]{\textcolor{purple}{JoakimA $\rightarrow$ AmitH: #1}}
    \newcommand{\jatoam}[1]{\textcolor{purple}{JoakimA $\rightarrow$ AmitM: #1}}

    \newcommand{\ah}[1]{\textcolor{orange}{AmitH: #1}}
    \newcommand{\ahtoja}[1]{\textcolor{orange}{AmitH $\rightarrow$ JoakimA: #1}}
    \newcommand{\ahtoam}[1]{\textcolor{orange}{AmitH $\rightarrow$ AmitM: #1}}

    \newcommand{\am}[1]{\textcolor{blue}{AmitM: #1}}
    \newcommand{\amtoah}[1]{\textcolor{blue}{AmitM $\rightarrow$ AmitH: #1}}
    \newcommand{\amtoja}[1]{\textcolor{blue}{AmitM $\rightarrow$ JoakimA: #1}}

    \newcommand{\as}[1]{\textcolor{olive}{AmitS: #1}}
\else
    \newcommand{\ja}[1]{}
    \newcommand{\jatoah}[1]{}
    \newcommand{\jatoam}[1]{}

    \newcommand{\ah}[1]{}
    \newcommand{\ahtoja}[1]{}
    \newcommand{\ahtoam}[1]{}

    \newcommand{\am}[1]{}
    \newcommand{\amtoah}[1]{}
    \newcommand{\amtoja}[1]{}

    \newcommand{\as}[1]{}
\fi

\newcommand{\eqref}[1]{(\ref{#1})}

\newcommand{\Real}{{\mathbb R}}

\newcommand{\Natural}{{\mathbb N}}
\newcommand{\SO}{{\mathrm{SO}}}

\newcommand{\Expect}{\mathbb{E}}
\newcommand{\Var}{\mathrm{Var}}
\newcommand{\Cov}{\mathrm{Cov}}

\newcommand{\eye}{\mathrm{I}}
\newcommand{\transp}{\mathrm{T}}
\newcommand{\frob}{\mathrm{F}}

\newcommand{\euler}{\mathrm{e}}

\newcommand{\argmin}[1]{\operatorname*{arg\min}_{#1}}

\newcommand{\volscalar}{x}
\newcommand{\vol}{{\bf \volscalar}}

\newcommand{\im}{{\bf y}}

\newcommand{\dist}{\nu}

\newcommand{\noisescalar}{\varepsilon}
\newcommand{\noise}{\boldsymbol{\noisescalar}}
\newcommand{\proj}{\mathrm{P}}

\newcommand{\ctf}{{\bf h}}

\newcommand{\rot}{R}
\newcommand{\M}{\mathcal{M}}

\newcommand{\mean}{\boldsymbol{\mu}}
\newcommand{\cov}{\Sigma}

\newcommand{\meanest}{\hat{\mean}}

\newcommand{\noisestd}{\sigma}
\newcommand{\imsize}{N}
\newcommand{\imsizesub}{\check \imsize}
\newcommand{\imdim}{\imsize^2}

\newcommand{\covest}{\hat{\cov}}

\newcommand{\neig}{q}

\newcommand{\eigvest}{\hat{\mathrm{V}}_\neig}

\newcommand{\coords}{\boldsymbol{\beta}}
\newcommand{\estcoords}{\boldsymbol{\hat{\beta}}}
\newcommand{\coordsupp}{\mathbf{B}}

\newcommand{\volest}{\hat{\vol}}

\newcommand{\I}{\mathrm{I}}
\newcommand{\W}{\mathrm{W}}
\newcommand{\D}{\mathrm{D}}
\renewcommand{\L}{\mathrm{L}}

\newcommand{\geig}{\boldsymbol{\phi} \kern-0.65em\hat{\phantom{\boldsymbol{\phi}}}\phantom{}} 
\newcommand{\geigs}{\hat{\phi}_s} 
\newcommand{\meig}{\phi} 
\newcommand{\dvolscalar}{{\alpha}}
\newcommand{\dvolscalarest}{{\hat{\alpha}}}
\newcommand{\dvol}{{\boldsymbol{\dvolscalar}}}
\newcommand{\dvolest}{{\hat{\boldsymbol{\alpha}}}}

\newcommand{\wtsrhsest}{{\bf b}}
\newcommand{\wtskerest}{\mathrm{K}}

\newcommand{\bigO}{O}

\newcommand{\bigOprob}{O_{\mathrm{P}}}

\newcommand{\vu}{{\bf u}}

\newtheorem{theorem}{Theorem}
\newtheorem{assumption}{Assumption}
\newtheorem{corollary}{Corollary}
\newtheorem{remark}{Remark}
\newtheorem{lemma}{Lemma}

\begin{document}

\title[Cryo-EM reconstruction of continuous heterogeneity by Laplacian spectral volumes]{Cryo-EM reconstruction of continuous heterogeneity by Laplacian spectral volumes}

\author{Amit Moscovich$^{1*}$, Amit Halevi$^{1*}$, Joakim And\'en$^2$ and Amit Singer$^{1,3}$}

\address{$^1$ Program in Applied \& Computational Mathematics, Princeton University, Princeton, NJ}

\address{$^2$ Center for Computational Mathematics, Flatiron Institute, New York, NY}

\address{$^3$ Department of Mathematics, Princeton University, Princeton, NJ}

\address{$^{*}$ Equal contribution.}

\eads{\mailto{amit@moscovich.org}, \mailto{ahalevi@princeton.edu}, \mailto{janden@flatironinstitute.org} and \mailto{amits@math.princeton.edu}}

\begin{abstract}
    Single-particle  electron cryomicroscopy is an essential tool for high-resolution 3D reconstruction of proteins and other biological macromolecules.
    An important challenge in cryo-EM is the reconstruction of non-rigid molecules with parts that move and deform.
    Traditional reconstruction methods fail in these cases, resulting in smeared reconstructions of the moving parts.
    This poses a major obstacle for structural biologists, who need high-resolution reconstructions of entire macromolecules, moving parts included.
    To address this challenge, we present a new method for the reconstruction of macromolecules exhibiting continuous heterogeneity.
    The proposed method uses projection images from multiple viewing directions to construct a graph Laplacian through which the manifold of three-dimensional conformations is analyzed.
    The 3D molecular structures are then expanded in a basis of Laplacian eigenvectors, using a novel generalized tomographic reconstruction algorithm to compute the expansion coefficients.
    These coefficients, which we name \textit{spectral volumes}, provide a high-resolution visualization of the molecular dynamics.
    We provide a theoretical analysis and evaluate the method empirically on several simulated data sets.
\end{abstract}    

\vspace{2pc}
\noindent{\it Keywords}: single particle electron cryomicroscopy, heterogeneity, tomographic reconstruction, molecular conformation space, manifold learning, Laplacian eigenmaps, diffusion maps 

\maketitle

\section{Introduction} \label{sec:Introduction}

The function of biological macromolecules is determined not only by their chemical composition but also by their 3D configuration.
Hence, accurately estimating these configurations is of great importance to the field of structural biology.
Macromolecules may deform their structure, resulting in a continuum of possible configurations, known as \textit{conformations}.
It is not always possible to isolate different conformations and study each separately.
As a result, practitioners often image a heterogeneous sample which is then ``purified'' computationally.

While X-ray crystallography and nuclear magnetic resonance (NMR) spectroscopy have been very successful in imaging homogeneous molecular structures to high resolution, they rely on aggregate measurements over an entire sample and are therefore ill-suited for imaging heterogeneous molecules.
Single-particle electron cryomicroscopy (cryo-EM), on the other hand, produces a separate image for each individual molecule, opening the door to accurate determination of structural variability.

The cryo-EM process consists of rapidly freezing a solution containing the molecular sample and then imaging it using a transmission electron microscope.
The electron detector captures a movie where each frame records the electron counts for every pixel.
To reduce the effects of ionization damage---which destroys the fine structure of the molecules---the electron dose is kept low, resulting in
exceptionally noisy images.
See bottom row in Figure \ref{fig:FakeKVorig} for examples.
Since each particle is randomly oriented with respect to the imaging plane, the resulting image contains projections of molecules from many random viewing directions.
Each projection image is typically modeled as the line integral of the particle's electric potential along the viewing direction, followed by convolution with a point spread function and the addition of noise \cite{Frank2006,VulovicEtAl2013}.
The goal of cryo-EM reconstruction is to invert this process and obtain an estimate of the molecular  volume from its projection images.
In recent years, better sample preparation techniques and improved detectors have led to reconstructions at a near-atomic resolution \cite{Kuhlbrandt2014,AmuntsEtal2014,LiaoEtal2013,BartesaghiEtal2018}.

The standard assumption in 3D reconstruction by cryo-EM is that of a homogeneous sample with no structural variability.
Many methods for homogeneous 3D reconstruction have been proposed \cite{Frank2006,BarnettGreengardPatakiSpivak2017,ChengEtal2015,MilneEtal2012,VinothkumarHenderson2016}.
The prevalent methods are based on a Bayesian approach \cite{Scheres2012a} which starts from some initial guess for the volume and then performs expectation-maximization (EM), alternating between estimating a latent distribution of viewing directions for every image and estimating the volume given these distributions.
As discussed above, however, the homogeneous assumption does not hold in general.
Resolving molecular structures with variability is known as the \textit{heterogeneity problem} in single-particle cryo-EM.
Two types of heterogeneity are typically considered: discrete and continuous.

\begin{figure}
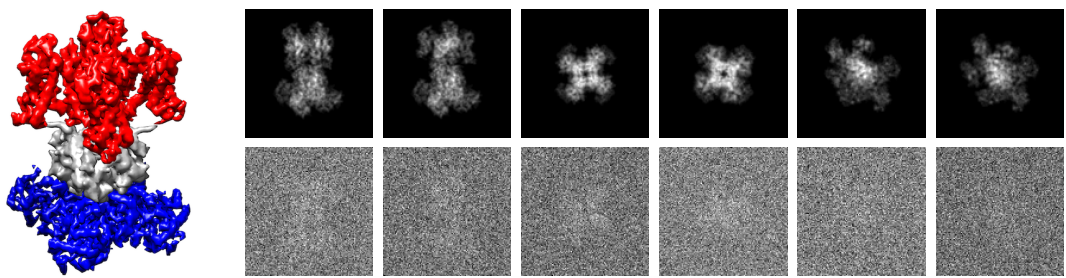

        \centering
        \qquad
        \quad\ 
        \begin{minipage}{0.20\linewidth}
            \includegraphics[height=36mm]{\detokenize{FakeKVorig}}
        \end{minipage}
        \begin{minipage}{0.7\linewidth}
            \includegraphics[width=17mm]{\detokenize{fig_im_clean_11}}
            \includegraphics[width=17mm]{\detokenize{fig_im_clean_21}}
            \includegraphics[width=17mm]{\detokenize{fig_im_clean_top_1}}
            \includegraphics[width=17mm]{\detokenize{fig_im_clean_top_2}}
            \includegraphics[width=17mm]{\detokenize{fig_im_clean_110}}
            \includegraphics[width=17mm]{\detokenize{fig_im_clean_210}}
                        
            \vspace{1.1mm}
            \includegraphics[width=17mm]{\detokenize{fig_im_noisy_11}}
            \includegraphics[width=17mm]{\detokenize{fig_im_noisy_21}}
            \includegraphics[width=17mm]{\detokenize{fig_im_noisy_top_1}}
            \includegraphics[width=17mm]{\detokenize{fig_im_noisy_top_2}}
            \includegraphics[width=17mm]{\detokenize{fig_im_noisy_110}}
            \includegraphics[width=17mm]{\detokenize{fig_im_noisy_210}}
        \end{minipage}
        \caption{\label{fig:FakeKVorig} \small
        The potassium ion channel used to simulate a heterogeneous molecular ensemble.  Note the $C_4$ rotational symmetry.
        (left) surface plot of the 3D density of a single molecule.
        We generated two synthetic datasets: \textsf{ChannelSpin} where the top red part is randomly rotated around the z axis (the molecule's axis of symmetry), and \textsf{ChannelStretch} where the bottom blue part is stretched along the x-y plane;
        (right) two different conformations from \textsf{ChannelSpin} projected along three orientations, from left to right: side view,  top view, and oblique view.
        The top row contains clean projections whereas the bottom row contains corresponding CTF-filtered projections with noise added.
        }
\end{figure}

\subsection{Discrete heterogeneity}
This is perhaps the simplest model for heterogeneity.
In this model, it is assumed that the particles in the sample can be approximated by a finite number of fixed volumes.
An example of a molecule that may be effectively modeled in this way is ATP synthase, an enzyme that acts as a molecular stepper motor and spends most of its time in one of three rotation angles \cite{YasudaEtal1998}.

Several software packages support reconstruction with discrete heterogeneity, also known as 3D classification in the cryo-EM community.
These include RELION \cite{Scheres2012b}, cryoSPARC \cite{PunjaniEtal2017}, FREALIGN/cisTEM \cite{LyumkisEtal2013,GrantRohouGrigorieff2018}, and EMAN2 \cite{TangEtal2007}.
To recover $K$ distinct volumes, these methods assign, for each particle image, a latent distribution over the set $\{1, \ldots, K\}$.
This is incorporated into the EM algorithm, similar to the latent distribution over the viewing directions.

\subsection{Continuous heterogeneity}

In this model, the molecular volumes in the sample vary continuously subject to the many constraints due to molecular bonds.
If  the number of degrees of freedom associated with the flexible motion is small then the space of molecular volumes forms a low-dimensional manifold (up to thermal vibrations).
Figure \ref{fig:FakeKVorig} shows a simple molecular model with continuous heterogeneity that we use in our simulations.
Here, the continuous motion is the free rotation of the top part around the vertical axis.
In this case, the manifold of molecular volumes is diffeomorphic to the unit circle $S^1$.

One approach for analyzing structural heterogeneity is to perform principal component analysis (PCA) of the 3D molecular structures represented as densities on an $\imsize \times \imsize \times \imsize$ voxel grid.
This idea goes back to \cite{LiuFrank1995} and was further developed by \cite{Penczek2002,PenczekEtal2006,PenczekFrankSpahn2006,PenczekKimmelSpahn2011,LiaoFrank2010}.
These methods estimate the $N^3 \times N^3$ covariance matrix of the 3D volumes and compute its leading eigenvectors, known as \textit{eigenvolumes}.
One variant relies on a consistent least-squares estimator for the covariance \cite{KatsevichKatsevichSinger2015,AndenKatsevichSinger2015,AndenSinger2018}.
These methods may capture continuous heterogeneity---as illustrated by \cite{AndenSinger2018}---but are ill-suited for high-resolution reconstruction, as we discuss in Section \ref{sec:LowResHeterogeneity}.
A notable exception is the method proposed in \cite{TagareEtal2015} that attempts to directly compute the leading eigenvectors, at high resolution, without estimating the entire covariance matrix.

A different approach is taken in \cite{DashtiEtAl2014,SchwanderFungOurmazd2014,FrankOurmazd2016,DashtiEtal2018} and is based on diffusion maps, a non-linear dimensionality reduction method that is well-suited for recovering low-dimensional manifold structure \cite{CoifmanEtal2005,CoifmanLafon2006}.
These methods first cluster the projection images by their viewing direction and then compute a separate low-dimensional embedding for each cluster.
All of these different embeddings are then aligned, yielding a global embedding of conformations.
Sets of close points in the global embedding may then be used to reconstruct a 3D volume corresponding to a particular conformation.
This approach faces two important challenges: first, unsupervised global registration of the embeddings is by itself a very challenging problem \cite{WangMahadevan2009,CuiChangShanChen2014}; second, each individual embedding uses only a small subset of images from a particular viewing direction, which may be insufficient for accurate manifold recovery.

The RELION software package has also been recently extended to include multi-body refinement \cite{NakaneEtal2018}.
This method takes a segmentation of a 3D molecular reconstruction and attempts to refine each part separately from a static base model, with independent viewing directions and shift parameters for each part.
Multi-body refinement, however, is limited to rigid variability and may fail to accurately reconstruct the interface between moving parts.

Other methods have been proposed based on normal mode analysis of the molecular structure reconstruction \cite{JinEtal2014,SchilbachEtal2017}.
However, the underlying harmonic oscillator model used in these methods may be too simple to describe sophisticated continuous variability such as structural deformations.
See \cite{SorzanoEtal2019} for a survey of methods for studying continuous heterogeneity using cryo-EM.

\subsection{Our contribution}

\begin{figure}
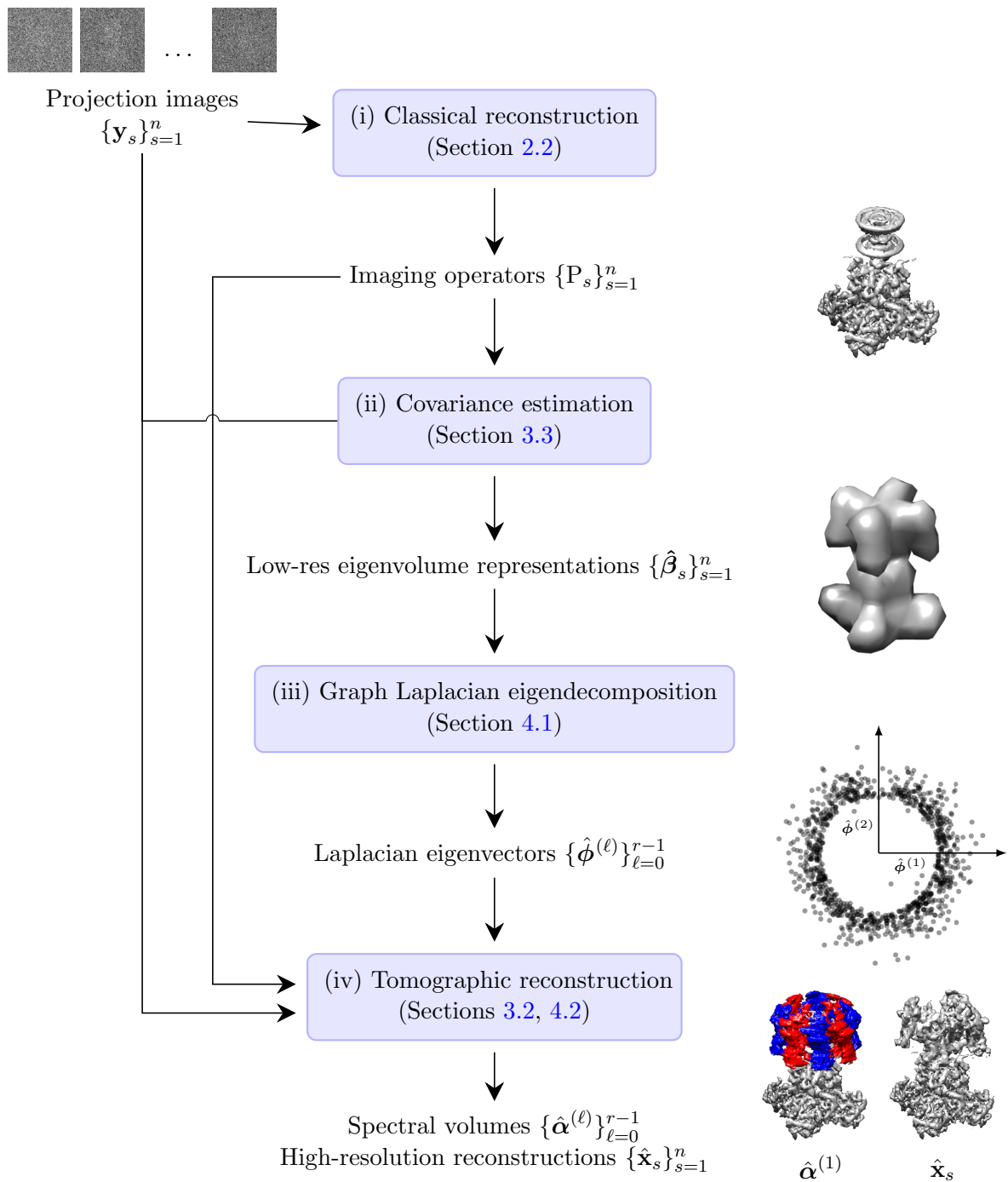
 \small
    \centering
    \begin{tikzpicture}[semithick,block/.style ={rectangle, draw=blue!30, thick, rounded corners, fill=blue!10,  align=center,  minimum height=3.5em},]
        \node (INPUT) [align=center]  at (-5.5,-2.0) { Projection images\\\( \{ \im_s \}_{s=1}^n \) };
        \node (INPUTIMG) at (-5.5,-1) { \includegraphics[width=10mm]{\detokenize{fig_im_noisy_110}}
            \includegraphics[width=10mm]{\detokenize{fig_im_noisy_21}}
            \(\begin{array}{c} \ldots \\ \ \end{array}\)
            \includegraphics[width=10mm]{\detokenize{fig_im_noisy_210}} };

        \node[block] (HOMORECONSTRUCTIONS) [align=center]  at (0,-2.25) {\ (i) Classical reconstruction \ \\(Section \ref{sec:InverseProblem})};
        \node (MEANVOL) [align=center]  at (6,-4.5) {\includegraphics[width=28mm]{\detokenize{spinning_alpha_1_pic}}};
    
        \node(IMAGINGOPERATORS) [align=center] at (0,-4.5) {Imaging operators \( \{ \proj_s \}_{s=1}^n
\)};
        \node(LEFTOFIMAGINGOPERATORS) at (-5,-4.5) {};

        \node (COVEST)[block] [align=center]  at (0,-6.75) {\ (ii) Covariance estimation \ \\(Section \ref{sec:LowResHeterogeneity})};

        \node (COVARIANCE) [align=center] at (0,-9) { Low-res eigenvolume representations \( \{ \estcoords_s \}_{s=1}^n \) };
        \node (LOWRESIMG) at (6,-9) { \includegraphics[width=24mm]{\detokenize{recon_cov_crop}} };
        \node (DMAPCOORDS) at (6,-13.5) { \includegraphics[width=40mm]{\detokenize{mani_example_circle}} };
        
        \draw[-latex] (6,-13.5)--(8,-13.5);
        \draw[-latex] (6,-13.5)--(6,-11.5);
        \node (PHI1) at (6.5,-13.7) {  \tiny \( \geig^{(1)} \) };
        \node (PHI2) at (5.68,-13.1) { \tiny \( \geig^{(2)} \) };

        \node[block,align=center] (GRAPHCOMPUTATIONS)  at (0,-11.25) {\ (iii) Graph Laplacian eigendecomposition \ \\(Section \ref{sec:GraphComputations})};

        \node[align=center] (LAPLACIANEIGENVECTORS)  at (0, -13.5) {Laplacian eigenvectors \( \{ \geig^{(\ell)} \}_{\ell=0}^{r-1}  \)  };

        \node[block,align=center] (RECONSTRUCTION)  at (0,-15.75) {\ (iv) Tomographic  reconstruction\ \\ (Sections \ref{sec:GenTomographicReconstruction},
\ref{sec:SpecVolEst})};

        \node (RECONSTRUCTEDVOLS) [align=center] at (0,-18) { Spectral volumes \( \{\dvolest^{(\ell)} \}_{\ell=0}^{r-1} \) \\ High-resolution reconstructions
\( \{ \volest_s \}_{s=1}^n \) };
        \node (HIGHRESIMG) at (5.0,-16.9) { \includegraphics[width=24mm]{\detokenize{spinning_alpha_3_pic}}};
        \node (HIGHRESRECON) at (7,-16.8) { \includegraphics[width=17.5mm]{\detokenize{recon_25000_r_15_crop}}};

        \node (LABELALPHA1) at (5.1,-18.4) {\( \dvolest^{(1)} \)};
        \node (LABELXRECON) at (7.0,-18.45) {\( \volest_s \)};
        
        \draw[->,shorten >=5pt] (INPUT) -- (HOMORECONSTRUCTIONS);
        \draw[->,shorten <=5pt] (HOMORECONSTRUCTIONS) -- (IMAGINGOPERATORS);
        \draw[->,shorten >=5pt] (IMAGINGOPERATORS) -- (COVEST);
        \path[->,shorten >=5pt,name path=line 1] (INPUT) |-  (COVEST);
        \draw[->,shorten <=5pt] (COVEST) -- (COVARIANCE);
        \draw[->,shorten >=5pt] (COVARIANCE) -- (GRAPHCOMPUTATIONS);
        \draw[->,shorten <=5pt] (GRAPHCOMPUTATIONS) -- (LAPLACIANEIGENVECTORS);
        \draw[->,shorten >=5pt] (LAPLACIANEIGENVECTORS) -- (RECONSTRUCTION);
        \draw[->] (INPUT) |- (-3.1,-16);
        \draw[->,shorten <=5pt] (RECONSTRUCTION) -- (RECONSTRUCTEDVOLS);
        \draw[->,name path=line 2] (IMAGINGOPERATORS) -- (-4.4,-4.5) |- (-3.1,-15.55);
        
        \def\radius{1.mm} 
        \path [name intersections={of = line 1 and line 2}];
        \coordinate (S)  at (intersection-1);
        \path[name path=circle] (S) circle(\radius);
        \path [name intersections={of = circle and line 1}];
        \coordinate (I1)  at (intersection-1);
        \coordinate (I2)  at (intersection-2);
        \tkzDrawArc[color=black](S,I1)(I2);
        \draw (INPUT) |- (I2);
        \draw (I1) -- (COVEST);
    \end{tikzpicture}
    \caption{\small High-level diagram of our method, illustrated on the \textsf{ChannelSpin} dataset. (i) Classical single-particle reconstruction, to obtain estimates of the CTF and viewing directions. (ii) Covariance estimation of the 3D density. The eigenvectors of the covariance matrix are then used to form a low-resolution 3D reconstruction from each projection image. (iii) Using the low-resolution reconstructions we build an affinity graph and compute its Laplacian eigenvectors.  (iv) We expand the unknown volumes in a basis of Laplacian eigenvectors and perform tomographic reconstruction. The result is  $r$ spectral volumes (left, overlaid on the mean image) which define  a high-resolution reconstruction for each projection image (right). \label{fig:pipeline}} 
\end{figure}

We present a new method for recovering continuous variability based on manifold learning.
In contrast to the viewing-direction specific manifold estimates of \cite{DashtiEtAl2014}, our method directly approximates the global manifold of conformations from all projection images, regardless of their viewing direction.

Throughout this paper, we identify molecular volumes with their electric potential sampled on a 3D voxel grid of dimension $\imsize^3$.
Under the continuous heterogeneity model, a single molecule corresponds to an embedded submanifold of $\Real^{\imsize^3}$. This manifold is the range of a smooth function that maps a set of conformation parameters to a volume.
A standard technique for approximating smooth functions on manifolds is by series expansion in Laplacian eigenfunctions.
This technique generalizes the familiar Fourier series expansion in Euclidean space.
However, to apply it we need to have the Laplacian eigenfunctions.
This is a ``chicken and egg'' problem: The computation of the Laplacian eigenfunctions requires the distribution of 3D volumes, which is the very thing we would like to estimate.
To resolve this problem, we use the covariance-based approach \cite{AndenSinger2018} to obtain low-resolution estimates of the 3D volumes.
These reconstructions are then used to form an empirical graph Laplacian whose $r$ eigenvectors with lowest eigenvalues \( \geig^{(0)},  \ldots, \geig^{(r-1)} \in \Real^n\) are used in lieu of the unknown Laplacian eigenfunctions.
Then we compute a set of expansion coefficient vectors \( \dvolest^{(0)}, \ldots, \dvolest^{(r-1)} \in \Real^{\imsize^3} \), which we refer to as \textit{spectral volumes}.
Together, they define a high-resolution 3D reconstruction $\volest_s$ for each projection image:
\begin{eqnarray}
    \volest_s := \sum_{\ell=0}^{r-1} \dvolest^{(\ell)} \geigs^{(\ell)},
\end{eqnarray}
To compute the expansion coefficients \( \{\dvolest^{(\ell)}\}_{\ell=0}^{r-1} \), we formulate a novel generalized tomographic reconstruction problem posed as a 3D deconvolution,
similar to \cite{WangShkolniskySinger2013}.
The convolution kernel is computed efficiently using a non-uniform fast Fourier transform (NUFFT) \cite{DuttRokhlin1993,GreengardLee2004} and
the solution is computed using the conjugate gradient method, leveraging the fast Fourier transform (FFT) for the application of the convolution.
These computational details are key for scaling up to high-resolution.
See Figure \ref{fig:pipeline} for a diagram of the main steps which constitute our method.
\begin{remark}
    The eigenvectors of the Laplacian can be used not only for function representation but also for non-linear dimensionality reduction (e.g. \cite{BelkinNiyogi2003,CoifmanEtal2005}).
    In our case, they define an embedding of the low-resolution reconstructions that is useful for visualizing the underlying manifold     of conformations.
    In Figure \ref{fig:ManiEmbeddings} we show the two-dimensional embedding of the {\normalfont \textsf{ChannelSpin}} dataset using the second and third eigenvectors. Despite the high noise levels, the underlying circular manifold of motions is recovered.
\end{remark}
\begin{remark}
    The spectral volumes have the same dimensionality as the high-resolution volumes that we reconstruct. They  may therefore be visualized as 3D molecular volumes, albeit with negative values as well as positive.
    These visualizations provide insight regarding the range of motions of the molecule.
    See Figure \ref{fig:RotSpecVols} for examples and Section \ref{sec:Theory} for an asymptotic analysis of the spectral volumes.
\end{remark}
\noindent Section \ref{sec:ProblemFormulation} defines the forward model and formulates the inverse problem for continuous heterogeneity in cryo-EM.
We describe our method in Section \ref{sec:Methods}, including the generalized tomographic reconstruction from noisy projection images.
Section \ref{sec:Algorithms} outlines the algorithms used and their computational complexity.
In Section \ref{sec:Theory} we prove the convergence of the spectral volumes and high-resolution reconstructions under the manifold assumption.
Finally, we present results on synthetic datasets in Section \ref{sec:Results}.

\begin{table} \small
    \caption{\small List of symbols. Scalars are denoted by italics, vectors by boldface letters, matrices by non-italicized capitals, estimators
are decorated with a hat.} \label{table:notation}
    \begin{center}
        \begin{tabular}{lll}
            \bf Name & \bf Domain & \bf Description\\\hline
            $n$ & $\Natural$ & Number of images and underlying molecular volumes \\
            $s$ & $1,\ldots,n$ & Index to molecular image/volume \\
            $\imsize$ & $\Natural$ & Image/volume size\\
            $\imsizesub$ & $\Natural$ & Downsampled image/volume size \\
            $\vol, \vol_s$ & $\Real^{\imsize^3}$ & Molecular volume\\
            $\volest_s$ & $\Real^{\imsize^3}$ & Our high-resolution molecular volume estimate \\
            ${\bf u}$ & $\{1,\ldots,N\}^3$ & Voxel index\\
            $\im, \im_s$ & $\Real^{\imsize^2}$ & Molecular image\\
            $\ctf, \ctf_s$ & $\Real^{\imsize^2}$ & Contrast transfer function (CTF)\\
            $\rot, \rot_s$ & $\SO(3)$ & 3D viewing orientation \\
            $\proj, \proj_s$ & $\Real^{\imsize^2 \times \imsize^3}$ & Imaging matrix (rotation, projection, and CTF) \\
            $\mathcal{F}_d$ &  & The $d$-dimensional discrete Fourier transform\\
            $\mean$ & $\Real^{\imsize^3}$ or $\Real^{\imsizesub^3}$ & Mean volume (high-res or low-res)\\
            $\cov$ & $\Real^{\imsizesub^3  \times \imsizesub^3}$ & Covariance matrix of downsampled molecular volumes\\
            $\neig$ & $\mathbb{N}$ & Number of PCA eigenvolumes\\
            $\eigvest$ & $\Real^{\imsizesub^3 \times \neig}$ & Eigenvolumes of the estimated covariance matrix\\
            $\coords(\vol), \coords_s$ & $\Real^q$ & PCA coordinates of a molecular volume\\
            $\coordsupp$ & $ \subseteq \Real^q $ & The domain of PCA coordinates\\
            $\dist(\coordsupp)$ & & Measure of volumes in PCA coordinate representation\\
            $\mathrm{W}$ & $\Real^{n \times n}$& Edge weights matrix\\
            $\L$ & $\Real^{n \times n}$& Graph Laplacian matrix\\
            $\M$ & $\subset \Real^{\imsize^3}$ & Riemannian submanifold of  molecular volumes\\
            $\meig^{(\ell)}$ &$\coordsupp \to \Real$& Laplace--Beltrami eigenfunction of the $\ell$\textsuperscript{th} smallest eigenvalue\\
            $\geig^{(\ell)}$ & $\Real^n$ & Laplacian eigenvector of the $\ell$\textsuperscript{th} smallest eigenvalue\\
            $r$ & $\Natural$ & Number of spectral volumes\\
            $\wtskerest$ & \( \mathbb{R}^{r\imsize^3 \times r \imsize^3}\) & Matrix of weighted projection-backprojections\\
            \( \wtsrhsest \) & \(\Real^{r \imsize^3}\) & Concatenation of weighted back-projection images\\
            $\dvol^{(\ell)}$ & $\Real^{\imsize^3}$ & Spectral volumes
        \end{tabular}
    \end{center}
\end{table}

\section{Problem formulation} \label{sec:ProblemFormulation}

We begin by describing the forward model for cryo-EM  and then define the inverse problem that we wish to solve, first by considering the simpler case without heterogeneity and then by generalizing to the case of continuous heterogeneity.

\subsection{Forward model} \label{sec:ForwardModel}
A sample of many identical molecules is prepared in a solution and then rapidly frozen, forming a thin sheet of vitreous ice which is then imaged using a transmission electron microscope.
The resulting image is a measurement of the electrostatic potential of this thin sheet, integrated along the direction perpendicular to the imaging plane.
The individual molecules, known as ``particles'' in the cryo-EM literature, are all captured in different orientations.

For every molecule in a particular 3D conformation, there is a corresponding real-valued electrostatic density map which we simply refer to as the {\em volume} and discretize it on an $\imsize \times \imsize \times \imsize$ grid of voxels.
We now describe the data generation model. First, the volumes $\vol_1,\ldots,\vol_n$ are drawn i.i.d. from some distribution on $\Real^{\imsize^3}$ which describes the structural variability of the molecule.
Then, linear imaging operators $\proj_1, \ldots, \proj_n \in \Real^{\imsize^2 \times \imsize^3}$ are drawn i.i.d from some distribution.
These operators are the composition of a volume rotation operator $\rot_s,$ tomographic projection onto the imaging plane, and convolution with a point spread function. 
 The individual particle images $\im_1,\ldots,\im_n \in \Real^{\imsize^2}$ are formed by
\begin{eqnarray} \label{eq:ForwardModel}
    \im_s = \proj_s \vol_s + \noise_s \quad \forall s=1,2,\ldots,n,
\end{eqnarray}
where $\noise_s$ are noise terms.
For simplicity, we assume that $\noise_s \sim \mathcal{N}(0,\sigma^2\eye_{\imsize \times \imsize})$.
The cryo-EM forward operator also includes an in-plane shift after the projection and filtering. In our pipeline, this is estimated and corrected for during the classical reconstruction stage (Figure \ref{fig:pipeline}, step (i)).

We consider the volumes $\vol_s \in \Real^{\imsize^3}$ as functions $\vol_s: M_\imsize^3 \to \Real$, where $M_{\imsize} := [-1, -1+2/\imsize, \ldots, 1 - 2/ \imsize]$, is the grid for even values of $\imsize$ (a similar grid may be defined for odd $\imsize$).
Similarly, the images $\im_s \in \Real^{\imsize^2}$ are functions $\im_s: M_{\imsize}^2 \to \Real$.
To define the imaging operators $\proj_s$ we must define the tomographic projection operation.
One approach to this is in terms of line integrals perpendicular to the projection plane but since the volumes lie on a discrete
grid one must incorporate an interpolation scheme.
An alternative is to express tomographic projection in the Fourier domain.
Let $\bf s$ be a $d$-dimensional signal on $M_{\imsize}^d$, its discrete Fourier transform (DFT) is given by
\begin{eqnarray}
    (\mathcal{F}_d {\bf s})({\bf k}) := \sum_{{\bf u} \in M_{\imsize}^d} e^{-2 \pi i \langle {\bf k}, \bf u \rangle} {\bf s}[{\bf u}]
    \quad {\forall\bf k} \in \Real^d
\end{eqnarray}
where $\bf k$ is a wave vector that corresponds to a particular directional frequency.
By the Fourier slice theorem, a tomographic projection along the $z$ axis in the spatial domain is equivalent to a restriction  to the $x$-$y$ plane in the Fourier domain \cite{Natterer2001}.
We use this fact to express the projection image \( \proj_s \vol_s\) in the Fourier domain as follows:
\begin{eqnarray} \label{eq:ImagingOperatorFourier}
    (\mathcal{F}_2 \proj_s \vol_s)([k_1,k_2]^\transp) = (\mathcal{F}_3\vol_s)(\rot_s^{-1}[k_1,k_2,0]^\transp)\cdot(\mathcal{F}_2 \ctf_s)([k_1,k_2]^\transp).
\end{eqnarray}
where $[k_1, k_2]$ is a wave vector in the resulting 2D projection image, $\rot_s \in \Real^{3 \times 3}$ is the rotation of particle number $s$ and $\ctf_s$ is the point-spread function whose Fourier transform \( \mathcal{F}_2 \ctf_s \) is known as the \textit{contrast transfer function} (CTF).
See Section 2 of \cite{AndenSinger2018} for more details on the forward model.

\subsection{Inverse problem} \label{sec:InverseProblem}

\textbf{Homogeneous case}. \
The traditional inverse problem in single-particle cryo-EM assumes that all of the molecular volumes in the sample are identical.
Thus, the forward model \eref{eq:ForwardModel} simplifies to
\begin{eqnarray}
    \im_s = \proj_s \mean + \noise_s \quad \forall s=1,2,\ldots,n,
\end{eqnarray}
where $\mean$ is a mean volume.
Suppose the orientations and CTFs are known so that we have the imaging operators $\proj_1, \ldots, \proj_n$.
Furthermore, suppose that the images are centered (i.e. in-plane shifts have been accounted for).
Then for a white Gaussian noise model, the maximum-likelihood estimate of $\mean$ is the solution to the following least-squares problem:
\begin{eqnarray}
    \label{eq:MeanLeastSquares}
    \meanest = \argmin{\mean \in \Real^{\imsize^3}} \sum_{s=1}^n \left\| \im_s - \proj_s \mean \right\|^2 \mbox{.}
\end{eqnarray}
This problem and regularized variants of it are not well-posed in general, with the condition number depending on the distribution of the viewing angles, the CTFs, and the desired resolution of the reconstruction.
Nevertheless, high accuracy solutions are routinely obtained using cryo-EM software packages. \cite{Scheres2012b,PunjaniEtal2017,GrantRohouGrigorieff2018,TangEtal2007}.\\[-0.5em]

\noindent \textbf{Continuous heterogeneity.}
Our main goal when analyzing a heterogeneous sample is to estimate the density of volumes $\vol \in \Real^{\imsize^3}$ associated with a given molecule.
We approach this problem by performing reconstructions of the individual volumes $\vol_1, \ldots, \vol_n$.
Clearly, estimating $n \imsize^3$ voxel values from merely $n \imsize^2$ noisy measurements is an ill-posed problem and much harder than the homogeneous problem, where only a single volume of $\imsize^3$ voxels needs to be estimated.
In this paper we make two main assumptions:
The first is that the molecular volumes in the sample lie near a low-dimensional manifold.  This model is natural since many heterogeneous macromolecules only have a few degrees of freedom that describe their range of motions \cite{DashtiEtAl2014,SchwanderFungOurmazd2014,FrankOurmazd2016,DashtiEtal2018}.
Varying these degrees of freedom traces out a smooth, low-dimensional manifold $\M \subset \Real^{\imsize^3}$.
The second assumption is that the imaging operators $\proj_s$ can be accurately estimated using standard cryo-EM reconstruction tools. This is the case when the molecule contains a large fixed component and a smaller heterogeneous part. A good indication that this is indeed the case for a particular dataset is when the reconstruction of the mean volume has a high resolution in some regions and lower resolution in others.

In the next section, we explain how we combine these assumptions with spectral techniques for function approximation on low-dimensional spaces to reconstruct all of the volumes in a heterogeneous molecular sample.

\section{Methods} \label{sec:Methods}

In this section, we describe our spectral approach to the reconstruction of molecular samples with continuous heterogeneity.
Our approach is based on the representation and approximation of molecular volumes using an orthogonal basis expansion of eigenfunctions. By expanding the molecular volumes in this basis and imposing the projection constraints we obtain a generalized spectral formulation of the cryo-EM reconstruction problem.

\subsection{Manifold spectral representation} \label{sec:SpectralGraphs}
Our method builds on the output of a low-resolution reconstruction method \cite{AndenSinger2018} that we describe in Section \ref{sec:LowResHeterogeneity}. In this method, each reconstructed volume is a linear combination of $q$ PCA eigenvolumes, hence it defines some mapping
\(
    (\im_s , \proj_s) \mapsto \coords_s
\)
where $\coords_s \in \Real^q$ is the vector of eigenvolume coefficients corresponding to a low-dimensional representation of $\vol_s$.
In what follows, we ignore potential ambiguities due to the projection and consider the low-resolution reconstruction as a linear dimensionality reduction of the underlying volume
\( \label{eq:MapToLowDim}
    \vol_s \mapsto \coords_s.
\)
Since we assumed the underlying manifold of volumes is $d$-dimensional, 
then if $d < q$ the image of this mapping is some compact domain $\coordsupp \subseteq \Real^q$ that is a $d$-dimensional immersed manifold.

In what follows we consider the approximation of smooth functions on general domains $\coordsupp$ via eigenfunctions of the Laplacian operator.
We briefly review some relevant facts \cite{GrebenkovNguyen2013}.
The Laplacian has a set of real eigenfunctions \( \meig^{(\ell)}: \coordsupp \to \Real \) that form a complete orthonormal basis of $L^2(\coordsupp)$ with corresponding non-negative eigenvalues $0 = \lambda_0 \le \lambda_1 \le \ldots \rightarrow \infty$.
The smoothness of $\meig^{(\ell)}$ is controlled by $\lambda_\ell$, which corresponds to the spatial frequency of $\meig^{(\ell)}$.
Consequently, the eigenfunctions with lowest eigenvalues form a natural basis for approximating smooth functions on $\coordsupp$.
In fact, this basis is optimal for the approximation of smooth functions with $L^2$ bounded gradient magnitudes \cite{AflaloBrezisKimmel2015}.
The idea of using Laplacian eigenfunctions for approximation and regression over arbitrary domains is a generalization of the classical approach for signal representation by Fourier series expansion \cite{GreblickiPawlak1985}.

Let us therefore consider the basis formed by the first $r$ eigenfunctions $\meig^{(0)}, \ldots, \meig^{(r-1)}$.
Fix a voxel $\vu \in \imsize^3$ and consider its associated restriction function $\vol[\vu]$.
We may approximate this function  using low-frequency eigenfunctions
\begin{eqnarray} \label{eq:LaplaceDecompVoxel}
    \vol[{\bf u}] \approx \sum_{\ell=0}^{r-1} \dvolscalar_{\bf u}^{(\ell)} \meig^{(\ell)}(\coords(\vol)),
\end{eqnarray}
where $\coords(\vol) \in \coordsupp$ is the image of $\vol$ in PCA coordinates.
This can be written more succinctly by aggregating  the coefficients for all voxels into a single volume, yielding
\begin{eqnarray} \label{eq:LaplaceDecomp}
    \vol \approx \sum_{\ell=0}^{r-1} \dvol^{(\ell)} \meig^{(\ell)}(\coords(\vol)), \quad \forall \coords \in \coordsupp.
\end{eqnarray}
We call the coefficient vectors $\dvol^{(0)}, \ldots, \dvol^{(r-1)} \in \Real^{\imsize^3}$ \emph{spectral volumes}.
Note that the above construction does not rely on a voxel-wise representation of the volumes as the same type of expansion can be done for volumes represented in any spatial basis.

The eigenfunctions are unknown, so we employ a widely used technique from the field of manifold learning, replacing them with estimates given by eigenvectors of a data-driven graph Laplacian.
More specifically, we build a weighted undirected graph, where the vertices correspond to the projection images $\im_1, \ldots, \im_n$ and the edge weights are estimates of the affinity between the underlying molecular conformations.
In our case, the affinities are computed from the low-resolution reconstruction coordinate $\estcoords_s$ described in Section \ref{sec:LowResHeterogeneity}. We then form the symmetric normalized graph Laplacian  and compute its $r$ eigenvectors with the lowest eigenvalues,
\begin{equation} \label{eq:defhatPhi}
    \geig^{(0)}, \ldots, \geig^{(r-1)} \in \Real^n.
\end{equation}
See Section \ref{sec:GraphComputations} for the specific algorithms used for forming the graph and computing these eigenvectors.
As we explain in Section \ref{sec:convergence}, we may assume that these estimates converge to the eigenfunctions in the sense that
\begin{eqnarray}
    \geigs^{(\ell)} \approx \frac{1}{\sqrt{n}} \meig^{(\ell)}(\coords_s) \quad \forall s = 1, 2, \ldots, n,
\end{eqnarray}
where the $\sqrt{n}$ factor is needed for proper normalization, so that
\begin{eqnarray}
    \sum_{s=1}^n \left(\geigs^{(\ell)} \right)^2 =1.
\end{eqnarray}
We can now write a data-driven variant of the spectral expansion in \eref{eq:LaplaceDecomp},
\begin{eqnarray} \label{eq:LaplaceDecompS}
\vol_s \approx \sqrt{n} \sum_{\ell=0}^{r-1} \dvol^{(\ell)} \geigs^{(\ell)} \quad \forall s = 1, 2, \ldots, n.
\end{eqnarray}
In the next section we explain how we estimate the coefficients of this expansion.

\subsection{Generalized tomographic reconstruction} \label{sec:GenTomographicReconstruction}

We assume that the molecular orientations can be accurately estimated using standard methods for homogeneous cryo-EM reconstruction \cite{PenczekKimmelSpahn2011,LiaoHashemFrank2015}, so that the projection operators $\proj_s$ are estimated to high accuracy.
By applying the imaging matrix $\proj_s$ to both sides of \eref{eq:LaplaceDecompS} and plugging in the forward model \eref{eq:ForwardModel}, we obtain
\begin{eqnarray}
    \im_s \approx \sqrt{n} \sum_{\ell=0}^{r-1} \left( \proj_s^{} \dvol^{(\ell)} \right) \geigs^{(\ell)} \quad \forall s = 1, 2, \ldots, n.
\end{eqnarray}
We seek spectral volumes that minimize the squared error
\begin{eqnarray} \label{eq:DiffVolDef}
    \left( \dvolest^{(0)}, \ldots, \dvolest^{(r-1)} \right)
    :=
    \argmin{\left\{\dvol^{(0)}, \ldots, \dvol^{(r-1)}\right\}}
    \sum_{s=1}^n
    \left\| \im_s - \sqrt{n}\sum_{\ell=0}^{r-1} \left( \proj_s^{} \dvol^{(\ell)} \right) \geigs^{(\ell)} \right\|^2.
\end{eqnarray}
The minimizer can be calculated efficiently by forming the normal equations and solving them using the conjugate gradient method.
See Section \ref{sec:SpecVolEst} for more details on the numerical solution of this minimization problem.
Note that in contrast to the low-resolution PCA eigenvolumes, the spectral volumes are at the full resolution $\imsize$.
Our high-resolution reconstructions of the molecular volumes are now given by
\begin{eqnarray} \label{eq:highresvolest}
    \volest_s = \sqrt{n}\sum_{\ell=0}^{r-1} \geigs^{(\ell)} \dvolest^{(\ell)} \quad \forall s = 1, 2, \ldots, n \mbox{.}
\end{eqnarray}
This estimator generalizes the least-squares estimator \eqref{eq:MeanLeastSquares} for a single mean volume to multiple volumes $\dvolest^{(0)}, \ldots, \dvolest^{(r-1)}$ whose contribution to the reconstructed volumes is given by the Laplacian eigenvectors $\geig^{(0)}, \ldots, \geig^{(r-1)}$ defined in Eq. \eqref{eq:defhatPhi}.

\subsection{Low-resolution reconstruction} \label{sec:LowResHeterogeneity}

While the approach outlined above provides a recipe for computing the eigenvectors $\geig^{(0)}, \ldots, \geig^{(r-1)}$ and using them to obtain high-resolution volume estimates, a crucial ingredient is missing still: the graph weights $W_{ij}$.
We would like them to approximate an affinity of the underlying molecular volumes.

Several approaches have been proposed for computing affinities between projection images of heterogeneous ensembles.
One of the earliest was to compute affinities using a common-line distance \cite{HermanKalinowski2008},
without estimating the relative orientations. This procedure finds the best common-line correspondence out of all candidate common lines, resulting in very noisy affinity estimates.
To reduce the noise one can first estimate the orientations of the projection images and then compute the common line distance based on the relative orientation.
This was proposed in \cite{ShatskyEtal2010}, however, the resulting affinity measure is still very noisy, so the authors first performed 2D class averaging within each set of projection images from the same orientation. However, this may average different conformations together.

We define the affinity $W_{ij}$ to be the Euclidean distance between the low-resolution reconstructions, obtained using the covariance estimation method
\cite{AndenSinger2018}.
This approach achieves robustness to noise without averaging different conformations together.
We now briefly describe their method.
The first step is to estimate the mean $\mean = \Expect[\vol]$ of the distribution of molecular volumes.  This is done by taking the derivative of \Eref{eq:MeanLeastSquares} with respect to $\mean$ and setting it equal to zero. This yields the normal equations
\begin{eqnarray}
    \label{eq:MeanNormal}
    \frac{1}{n} \left( \sum_{s=1}^{n} \proj_s^\transp \proj_s^{} \right) \meanest = \frac{1}{n} \sum_{s=1}^{n} \proj_s^\transp \im_s^{} \mbox{.}
\end{eqnarray}
This formulation  corresponds to the maximum-likelihood estimator of $\Expect[\vol]$ in the setting of Gaussian white noise.
As a consequence, $\meanest$ is a consistent estimator \cite{KatsevichKatsevichSinger2015}.
A similar estimator for the covariance matrix
\(
    \Cov[\vol] := \Expect[(\vol-\Expect[\vol])(\vol-\Expect[\vol])^\transp]
\)
is given by
\begin{eqnarray}
    \covest
    =
    \argmin{\cov \in \Real^{\imsize^3 \times \imsize^3}}
    \sum_{s=1}^n
    \left \|
        (\proj_s^{} \cov \proj_s^\transp + \noisestd^2 \eye_{\imdim})
        -
        (\im_s - \proj_s \meanest)(\im_s - \proj_s \meanest)^\transp
    \right \|^2_\frob \mbox{.}
\end{eqnarray}
While not a maximum-likelihood estimator, it is consistent under mild conditions  \cite{KatsevichKatsevichSinger2015}.
Computing its normal equations yields a linear system in $\bigO(\imsize^6)$ variables.
Fortunately, this linear system can be reformulated as a deconvolution problem in six dimensions.
Precalculating the convolution kernel requires $\bigO(\imsize^6 \log \imsize + n \imsize^4)$ operations, but it can then be applied with complexity $\bigO(\imsize^6 \log \imsize)$.
The equations can now be solved using the preconditioned conjugate gradient method.
Empirically, it takes around $50$ iterations to converge \cite{AndenSinger2018}.

While more efficient than a naive approach, the algorithm outlined above still scales poorly in image size $\imsize$.
As a result, this covariance estimation method is not currently practical for $\imsize > 25$.
Furthermore, from a simple dimensionality argument, to estimate the $\bigO(\imsize^6)$ elements of $\Cov[\vol]$ from $n$ images of size $\imsize \times \imsize$, we need at least $n = \bigO(\imsize^6 / \imsize ^2) = O(\imsize^4)$ images.
So to apply the algorithm to experimental data, we must first downsample the images from $\imsize \times \imsize$ to $\imsizesub \times \imsizesub$.
It is possible to gain insight on the structural variability using this approach, but the resulting reconstructions are of low-resolution.

After obtaining the mean and covariance estimates $\meanest$ and $\covest$, the volumes $\vol_1, \ldots, \vol_n$ can be reconstructed by the PCA method introduced in \cite{PenczekKimmelSpahn2011}.
First, the $\neig$ eigenvectors, or eigenvolumes, of $\covest$ are extracted and arranged as columns in a $\imsizesub^3 \times \neig$ matrix $\eigvest$.
They represent the principal directions of molecular volume variability in $\Real^{\imsizesub^3}$.
Together with the estimated mean, they define an affine $\neig$-dimensional subspace of $\Real^{\imsizesub^3}$ of the form $\meanest + \eigvest \coords$, where $\coords \in \coordsupp \subseteq\Real^\neig$ is a coordinates vector.
Each image $\im_s$ may then be associated with a volume in the affine subspace through \cite{AndenSinger2018}
\begin{eqnarray} \label{eq:CovCoords}
    \estcoords_s := \argmin{\coords \in \Real^q} \frac{1}{\noisestd^2} \left\| \im_s - \proj_s \left( \meanest+\eigvest \coords \right) \right\|^2 + \left\|\Lambda_q^{-1/2} \coords \right\|^2 \mbox{,}
\end{eqnarray}
where $\Lambda_q^{} = \eigvest^\transp \covest \eigvest^{}$ is the diagonal matrix of the leading $q$ eigenvalues of $\covest$.
The above estimator is the maximum a posteriori (MAP) estimator of the coordinates of $\vol_s$ for Gaussian distributions of $\vol_s$ and $\noise_s$.
It is also equal to the Wiener filter estimator and the linear minimum mean squared error estimator of the coordinates \cite{Mallat2009,Kay1993}.

Given the solutions to \eref{eq:CovCoords}, we have a low-resolution estimate of each volume $\vol_s$ given by $\meanest + \eigvest \estcoords_s$.
We assume that the manifold structure of $\M$  is not destroyed by the mapping of projection images to coordinate vectors in $\Real^\neig$, hence that it is possible  to invert this process and associate a unique molecular conformation with every low-dimensional reconstruction.
If the intrinsic dimensionality of the conformation space is low and the volumes vary smoothly along this space then the inverse map $\coordsupp \to \Real^{\imsize^3}$ can be approximated by a small number of spectral volumes.

\section{Algorithms and computational complexity} \label{sec:Algorithms}

In this section, we provide the technical details of our reconstruction method.
In Section \ref{sec:GraphComputations} we describe the precise methods used to form the graph Laplacian and compute its eigenvectors, and in Section \ref{sec:SpecVolEst} we describe the deconvolution-based solution of the generalized tomographic reconstruction problem \eref{eq:DiffVolDef}.

\subsection{Graph computations} \label{sec:GraphComputations}

To compute the PCA eigenvolumes, we begin by downsampling the input images to size $\imsizesub \times \imsizesub$, where $\imsizesub$ is typically about 16.
These images are then fed into the mean and covariance estimation pipeline described in \cite{AndenSinger2018}.
It has computational complexity $\bigO(n \imsizesub^4 + \sqrt{\kappa^\prime} \imsizesub^6 \log \imsizesub)$. The condition number $\kappa^\prime$ is of the order of $100$.
The top $q$ eigenvectors of the estimated covariance $\covest$ are computed and the $q$-dimensional coordinates $\estcoords_s$ of each image are obtained via \eref{eq:CovCoords}.
This step has computational complexity $\bigO(q \imsizesub^3 \log \imsizesub + n q^2 \imsizesub^2)$, following the algorithm described in \cite{AndenSinger2018}.
A weighted undirected graph is then constructed with vertices that correspond to the images \( \im_1, \ldots, \im_n \) and edge weights calculated from the PCA coordinates $\estcoords_1, \ldots, \estcoords_n$.
We tested two kinds of weight matrices:
\begin{enumerate}
    \item Gaussian kernel weights \(W_{ij} = \euler^{-\| \estcoords_i - \estcoords_j \|^2/2\sigma^2} \).
    \item Binary symmetric KNN matrices, whereby \( W_{ij} = 1 \) if and only if $\estcoords_i$ is one of the $k$ nearest neighbors of $\estcoords_j$ or vice versa, and \( W_{ij} = 0 \) otherwise.
\end{enumerate}
In our preliminary experiments we obtained similar results with both choices.
For our final results, we chose to use the symmetric KNN graph since it is sparse, which reduces the memory and computational costs.
For the Laplacian matrix, we use the symmetric normalized graph Laplacian
\begin{eqnarray}
    \L := \D^{-1/2} (\D - \W) \D^{-1/2} = \I - \D^{-1/2} \W \D^{-1/2},
\end{eqnarray}
where $\D$ is a diagonal matrix that satisfies \( D_{ii} = \sum_j W_{ij} \).
The symmetry of $\L$ permits the use of specialized algorithms for eigenvector calculation and guarantees that the resulting eigenvectors are orthogonal.
See the tutorial by \cite{Vonluxburg2007} for other common choices of weight and Laplacian matrices.

We build the KNN weights matrix $\W$ using MATLAB's \textsf{knnsearch} function which for low dimensions is based on a KDTree \cite{FriedmanBentleyFinkel1977}.
The running time of this part is \( \bigO(q n \log n)\) where $q$ is the dimension of the PCA coordinates $\estcoords_s$ used in the low-resolution reconstruction.
We then form the Laplacian matrix $\L$ and compute its $r$ eigenvectors $\geig^{(0)}, \ldots, \geig^{(r-1)}$ with lowest eigenvalues using MATLAB's \textsf{eigs} function.
This function implements the Krylov--Schur algorithm \cite{Stewart2001}.
The matrices $\W$ and $\L$ are stored as sparse matrices of average degree $\bigO(k)$, hence their memory usage is $\bigO(nk)$.
There exist newer methods of computing eigenvectors, such as the algebraic multigrid preconditioner used by the \textsf{megaman} manifold learning package \cite{McQueenEtal2016,OlsonSchroder2018}.
We did not incorporate such methods in the current work, as the eigenvector calculation step was not a bottleneck in our implementation.

\subsection{Spectral volume estimation} \label{sec:SpecVolEst}

Recall that the spectral volumes are defined in \eref{eq:DiffVolDef} as minimizers of the generalized tomographic reconstruction equation, \eref{eq:DiffVolDef}.
To find this minimum, we compute the gradient with respect to $\{\dvol^{(\ell)}\}_{\ell=0}^{r-1}$ and set it to zero, obtaining the normal equations
\begin{eqnarray} \label{eq:DiffvolEquality}
    \hspace{-10mm}
    \frac{1}{\sqrt{n}}
    \sum_{s=1}^n \geigs^{(\ell)} \proj_s^\transp\im_s^{}
    =
    \sum_{m=0}^{r-1}
    \sum_{s=1}^n
    \geigs^{(\ell)}
    \geigs^{(m)}
    \proj_s^\transp \proj_s^{}
    \dvol^{(m)}
    \quad
    \forall \ell = 0, 1, \ldots, r-1.
\end{eqnarray}
We can rewrite the equation  in vector notation by defining the vectors $\wtsrhsest^{(0)}, \ldots, \wtsrhsest^{(r-1)} \in \mathbb{R}^{\imsize^3}$ to be weighted backprojected images
\begin{eqnarray}
    \wtsrhsest^{(\ell)} = \frac{1}{\sqrt{n}}\sum_{s=1}^n \geigs^{(\ell)} \proj_s^\transp \im_s^{} \ {},
\end{eqnarray}
and $\wtskerest \in \mathbb{R}^{r\imsize^3 \times r \imsize^3}$ to be an $r \times r$ block matrix, with blocks of size $\imsize^3 \times \imsize^3$.  Each block is a weighted sum of projection-backprojection matrices, with its $(\ell, m)$ block given by
\begin{eqnarray} \label{eq:kermat}
    \wtskerest^{(\ell, m)}
    =
    \sum_{s=1}^n
    \geigs^{(\ell)}
    \geigs^{(m)}
    \proj_s^\transp
    \proj_s^{} .
\end{eqnarray}
By defining the vector $\wtsrhsest \in \mathbb{R}^{r\imsize^3}$ to be the concatenation of  $\wtsrhsest^{(0)}, \ldots, \wtsrhsest^{(r-1)}$ and $\dvol \in \mathbb{R}^{ r \imsize^3}$ to be the concatenation of  \(\dvol^{(0)}, \ldots, \dvol^{(r-1)}\) we can rewrite \eref{eq:DiffvolEquality} as
\begin{eqnarray} \label{eq:DiffvolEquation}
    \wtsrhsest = \wtskerest \dvol .
\end{eqnarray}
Since $\wtskerest$ is of size $\imsize^3 r \times \imsize^3 r$, it would be very expensive to directly solve this equation using standard direct inversion algorithms such as those based on LU or Cholesky decomposition, since this would require $\bigO(\imsize^9 r^3)$ operations.
Even merely storing the matrix $\wtskerest$ in RAM may be prohibitive.
However, if we use an iterative solver such as the conjugate gradient method, we do not need to explicitly store the matrix $\wtskerest$ so long as we have an efficient method to apply it.
To this end, we draw on the work of \cite{WangShkolniskySinger2013} and note that applying $\proj_s^\transp \proj_s^{}$ to a volume is equivalent to convolving that volume with a kernel calculated from $\rot_s$ and $\ctf_s$.
A complication arises from the fact that the points $\rot_s^{-1}[k_1,k_2,0]^\transp$ in \eqref{eq:ImagingOperatorFourier} do not lie on a regular grid, hence
to evaluate the expression \( (\mathcal{F}_3\vol_s)(\rot_s^{-1}[k_1,k_2,0]^\transp) \) we need to compute Fourier amplitudes on a non-regular
grid which cannot be achieved through  the standard FFT. Instead, we use the \textsf{FINUFFT} non-uniform fast Fourier transform software package \cite{BarnettEtal2019}.
It has computational complexity $\bigO(\imsize^3 \log \imsize + S)$ where $S$ is the number of points at which the transform is computed.
Here, $S = \imsize^2 n$, as both $\ctf_s$ and $\im_s$ are of size $\imsize \times \imsize$, and we consider $n$ instances of projection images.
We must compute the convolution kernel that corresponds to  $\wtskerest^{(\ell, m)}$ for each of the $r^2$ $(\ell,m)$-pairs, and $\wtsrhsest^{(\ell)}$ for each $\ell$.
Thus, the total time to calculate the convolution kernels of all the blocks of $\wtskerest$ is $\bigO(r^2\imsize^3 \log \imsize + r^2n\imsize^2)$.
The backprojected images  vector $\wtsrhsest$ is also calculated from $\rot_s$, $\ctf_s$, and $\im_s$ using a non-uniform FFT at a total computational cost of $\bigO(r\imsize^3 \log \imsize + rn\imsize^2)$.

Each step of the conjugate gradient method involves applying the forward operator $\wtskerest$ as well as performing several vector dot products and additions.
Applying the forward operator is done using $r^2$ FFT operations of size $\imsize \times \imsize \times \imsize$, which has a total complexity of $\bigO(r^2 \imsize^3 \log \imsize)$.
The complexity of the conjugate gradient method is thus $\bigO(\sqrt{\kappa} r^2 \imsize^3 \log \imsize)$, where $\kappa$ is the condition number of $\wtskerest$, since the conjugate gradient method converges in $\bigO(\sqrt{\kappa})$ steps \cite{GolubVanLoan2013,TrefethenBau1997}.
In conclusion, the total runtime for solving the normal equations \eqref{eq:DiffvolEquality} is $\bigO(r^2 n \imsize^2 + \sqrt{\kappa} r^2 \imsize^3 \log \imsize)$.
For our synthetic data sets \textsf{ChannelSpin} and \textsf{ChannelStretch}, using $r=15$ spectral volumes we found that $\kappa$ is of the order of $10$--$30$.
See Section \ref{sec:Runtime} for empirical runtimes on these data sets.
\begin{remark}
    The running time may be reduced by computing an approximation to $\wtskerest$.
    In the proof of Theorem \ref{thm:specvolconvergence} we show that \( \wtskerest^{(\ell, m)} \to \delta_{\ell,m}\Expect[\proj^\transp \proj^{}] \) in probability. We can thus approximate $\wtskerest$ by setting the off-diagonal blocks to zero and setting the diagonal blocks to the  empirical estimate of \(\Expect[\proj^\transp \proj^{}]\)
    \begin{eqnarray}
        \wtskerest^{(\ell, \ell)}
        =
        \frac{1}{n}
        \sum_{s=1}^n
        \proj_s^\transp
        \proj_s^{}.
    \end{eqnarray}
    With this approximation, the time to approximate $\wtskerest$ reduces to $\bigO(\imsize^3 \log \imsize + n\imsize^2)$ and is now dominated by the computation of $\wtsrhsest$. The time to multiply vectors by $\wtskerest$ is \( \bigO(r\imsize^3 \log \imsize) \),
    so the total runtime drops by a factor of $r$ to
    \(
        \bigO(rn\imsize^2 + \sqrt{\kappa} r \imsize^3 \log \imsize)
    \).
\end{remark}
\section{Theory} \label{sec:Theory}
In this section, we analyze the solution to the generalized tomographic reconstruction as defined in \eref{eq:DiffVolDef}, starting with a simplified special case.

\subsection{Warmup: Spectral volumes without projections} \label{sec:TheoryNoProj}
We first analyze the solution in an easy setting where the imaging operators  $\proj_1, \ldots, \proj_n$ are all equal to the identity matrix.
That is, we have direct, albeit noisy, measurements ${\bf z}_s = \vol_s + \noise_s$ without projections and point spread function.
This case is directly applicable for reconstructing a manifold of 2D images, as we later demonstrate in Section \ref{sec:2DClock}.
In this setting, the spectral volumes $\dvolest^{(0)}, \ldots, \dvolest^{(r-1)}$ minimize
\begin{eqnarray} \label{eq:DiffvolNoProj}
    \sum_{s=1}^n \left\| {\bf z}_s - \sqrt{n}\sum_{\ell=0}^{r-1} \geigs^{(\ell)} \dvol^{(\ell)}  \right\|^2 \mbox{.}
\end{eqnarray}
In this sum, each voxel $\bf u$ can be considered separately, giving
\begin{eqnarray}
    \dvolscalarest^{(\ell)}[{\bf u}] = \argmin{\dvolscalar^{(\ell)}[\bf u]} \sum_{s=1}^n \left| z_s[{\bf u}] - \sqrt{n}\sum_{\ell=0}^{r-1} \geigs^{(\ell)} \dvolscalar^{(\ell)}[{\bf
u}] \right| ^2 \mbox{.}
\end{eqnarray}
For a symmetric graph Laplacian $\L$, the eigenvectors $\geig_0, \ldots, \geig_{r-1}$ form an orthonormal set.
Hence the coefficient $\dvolscalarest^{(\ell)}[{\bf u}]$ is given by an orthonormal projection of $z[{\bf u}]$ onto $\geig^{(\ell)}$
\begin{eqnarray}
    \dvolscalarest^{(\ell)}[{\bf u}] = \frac{1}{\sqrt{n}}\sum_{s=1}^{n} \geigs^{(\ell)} z_s^{}[{\bf u}] = \frac{1}{\sqrt{n}}\sum_{s=1}^{n} \geigs^{(\ell)} (\volscalar_s^{}[{\bf u}] + \noisescalar_s^{}[{\bf u}]),
\end{eqnarray}
or, in vector form,
\begin{eqnarray} \label{eq:SpecVolsSolWithoutProj}
    \hspace{-10mm}
    \dvolest^{(\ell)}
    =
    \frac{1}{\sqrt{n}}\sum_{s=1}^{n} \geigs^{(\ell)} \left( \vol_s^{} + \noise_s^{}\right)
    =
    \frac{1}{\sqrt{n}}\sum_{s=1}^{n} \geigs^{(\ell)} \vol_s^{} + \mathcal{N}\left(0,  \frac{\sigma^2}{n} \eye_{\imsize^2}\right).
\end{eqnarray}
The last equality stems from the fact that the noise terms satisfy \( \noise_s \sim \mathcal{N}(0, \sigma^2 \eye_{\imsize^2}) \).
Consequently, the spectral volumes in this simplified model are, up to a noise term, orthogonal projections of the true volumes $\vol_1, \ldots, \vol_n$ onto the basis of Laplacian eigenvectors.
In the next subsection, we show that this is also the case when tomographic projections are incorporated into the model.

\subsection{Spectral volumes with projections}

We now consider the full forward model with non-trivial imaging operators $\proj_1, \ldots, \proj_n$.
First note that in our model, the images \( \im_1, \ldots, \im_n \) and the imaging operators are random vectors, therefore the Laplacian eigenvectors \(\geig^{(0)} \ldots \geig^{(r-1)} \in \Real^n\) are also random vectors.
For our analysis, we make the following two assumptions:
\begin{assumption} \label{assump:IndependenceDiffcoordsProjections}
    Let $\im_s$ be an image, drawn according to the forward model \eref{eq:ForwardModel}.
    Then its  Laplacian eigenvector coordinates $\geigs^{(0)}, \ldots, \geigs^{(r-1)}$ are independent of $\proj_s$.
\end{assumption}
In other words, the Laplacian eigenmap (or diffusion map) coordinates are independent of the viewing direction and CTF of the particle.
We can justify this assumption by assuming that the covariance-based method of \cite{AndenSinger2018} performs accurate low-resolution reconstruction, regardless of the viewing angle.
\begin{assumption} \label{assump:EigenvectorBound}
    For any $r > 0$, the following sum converges in probability:
    \begin{eqnarray} \label{eq:FourthMomentConvergence}
        \max_{\ell \in \{ 0, \ldots, r-1 \}}
        \sum_{s=1}^n
        \left( \geigs^{(\ell)} \right)^4
        \to
        0.
    \end{eqnarray}
    That is, for any $\epsilon > 0$, the probability that \(  \sum_{s=1}^n \left(\geigs^{(\ell)} \right)^4 > \epsilon \) tends to zero as \( n \to \infty \).
\end{assumption}
Note that from the normalization constraint $\sum_{s=1}^n ( \geigs^{(\ell)}
)^2 = 1$, unless the energy of the eigenvectors is highly concentrated, we expect to have $\geigs^{(\ell)} \sim 1/\sqrt{n}$ and thus \( \sum_{s=1}^n ( \geigs^{(\ell)}
)^4 \sim 1/n \), in which case Assumption \ref{assump:EigenvectorBound}  holds. 
Under standard assumptions the Laplacian eigenvectors converge to limiting eigenfunctions of some differential operator.
As we show in Section \ref{sec:convergence}, if the eigenfunctions are bounded and this spectral convergence holds then Assumption \ref{assump:EigenvectorBound} follows.

Before stating our main result, we recall big-O in probability notation for stochastic boundedness: a sequence of random variables $\{X_n\}_{n=1}^\infty$ satisfies $X_n = \bigOprob(f(n))$ if for every $\epsilon > 0$, there is some bound $M_\epsilon$ such that $\Pr[|X_n|/f(n) > M_\epsilon] < \epsilon$.
We now state our main result which characterizes the estimated spectral volumes $\dvolest^{(0)}, \ldots, \dvolest^{(r-1)}$ up to a stochastically bounded error.
\begin{theorem}(Spectral volume convergence) \label{thm:specvolconvergence}
    Let $\geig^{(\ell)}$ be an eigenvector of the symmetric graph Laplacian  described in Section \ref{sec:GraphComputations}.
    Under Assumptions \ref{assump:IndependenceDiffcoordsProjections} and \ref{assump:EigenvectorBound}, the spectral volumes as defined in \eref{eq:DiffVolDef} satisfy
    \begin{eqnarray} \label{eq:SpecvolConvergence}
        \dvolest^{(\ell)}
        =
        \Expect
        \left[
            \frac{1}{\sqrt{n}}
            \sum_{s=1}^n
            \geigs^{(\ell)} \vol_s^{}
        \right]
        +
        \bigOprob\left(\frac{1}{\sqrt{n}}\right),
    \end{eqnarray}
    where the expectation is taken with respect to the random draw of projection images as described in Section \ref{sec:ForwardModel}.
\end{theorem}
The proof is in \ref{sec:proofs}.

\subsection{Convergence of the reconstructed volumes} \label{sec:convergence}
Consider the graph Laplacian eigenvectors \( \{ \geig^{(\ell)} \}_{\ell=0}^{n-1}\) computed from  the low-resolution reconstruction coordinates \(\estcoords_1, \ldots, \estcoords_n  \in \coordsupp\). Several variants of the discrete graph Laplacian are known to converge to a continuous linear operator on \( \coordsupp \).
This convergence is not only pointwise, but also spectral, meaning that the eigenvectors of the graph Laplacian converge to the eigenfunctions \( \meig^{(\ell)} \) of this operator \cite{VonluxburgBelkinBousquet2008,RosascoBelkinDevito2010,LeeIzbicki2016}.
In particular cases, the limiting operator  is the continuous Laplacian, but more generally it is a weighted Laplacian operator, or Fokker--Planck operator, which  has an additional drift term towards, or away from, high-density regions \cite{CoifmanLafon2006,NadlerLafonCoifmanKevrekidis2005,TingHuangJordan2010}.
For our theoretical analysis we only need spectral convergence towards some set of eigenfunctions, not necessarily the Laplacian eigenfunctions.
We formulate this requirement in the following assumption.
For simplicity, we ignore possible eigenvalue multiplicities.
\begin{assumption} \label{assump:eigenconvergence}
    The domain $\coordsupp$ is compact and the eigenvectors \( \geig^{(0)}, \geig^{(1)}, \ldots \in \Real^{n} \) of the graph Laplacian, ordered by their eigenvalues, converge in probability to a set of  eigenfunctions  \( \meig^{(0)}, \meig^{(1)}, \ldots:\coordsupp \to \Real   \) of some continuous linear differential operator on $\coordsupp$, in the sense that
    \begin{eqnarray} \label{eq:eigenconvergence}
        \sup_{s=1,\ldots,n}
        | \sqrt{n}\geigs^{(\ell)} - \meig^{(\ell)}(\estcoords_s)| \to 0.
    \end{eqnarray}
    Furthermore, \( \{ \meig^{(\ell)} \} \) form an orthonormal set with respect to the measure $\dist(\coordsupp)$,
    \begin{eqnarray}
        \langle \meig^{(\ell)}, \meig^{(m)} \rangle
        =
        \int_{\coordsupp} \meig^{(\ell)}(\coords) \meig^{(m)}(\coords) \mathrm{d}\dist(\coords) = \delta_{\ell,m}.
    \end{eqnarray}
\end{assumption}
\begin{remark}
    Under this assumption, the eigenfunctions \( \{ \meig^{(\ell)} \} \) have upper bounds, which we denote as $U_\ell$.
    This is due to the fact that they are continuous functions on a compact domain.
    It follows that,
    \begin{equation}
        \sum_{s=1}^n \left( \geigs^{(\ell)} \right)^4
        \to
        \sum_{s=1}^n \left( \frac{1}{\sqrt{n}} \meig^{(\ell)}(\estcoords_s)
\right)^4
        \le
        \frac{1}{n} U_\ell^4.
    \end{equation}
    Thus, Assumption \ref{assump:EigenvectorBound} follows from Assumption \ref{assump:eigenconvergence}.
\end{remark}
\begin{remark}
    The $\sqrt{n}$ term in \eqref{eq:eigenconvergence} is necessary for the eigenvector normalization, since
    \begin{eqnarray}
        \sum_{s=1}^n \left( \geigs^{(\ell)} \right)^2
        \to
        \sum_{s=1}^n \left( \frac{1}{\sqrt{n}}\meig^{(\ell)}(\estcoords_s) \right)^2
        \to
        \int_\coordsupp \left( \meig^{(\ell)}(\coords) \right)^2 \mathrm{d}\dist(\coords) = 1.
    \end{eqnarray}
\end{remark}
\noindent By Propositions 2.1 and 2.2 of \cite{KatsevichKatsevichSinger2015}, the low-resolution mean and covariance estimates are consistent.
However, unlike these aggregate quantities, the PCA coordinates $\estcoords_s$ are computed from a single image, so they must contain an irreducible error term due to the finite noise level.
We codify this in the following assumption.
\begin{assumption} \label{assump:PCACoordsConvergence}
    The estimated PCA coordinates are correct up to some stochastically bounded noise term, 
    \begin{eqnarray} \label{eq:LowResReconsError}
        \estcoords_s =  \coords(\vol_s) + O_P(1).
    \end{eqnarray}
\end{assumption}
We now show that, up to noise, the spectral volumes are merely voxel-wise orthogonal projections of the
true volumes $\vol_1, \ldots, \vol_n$  onto a basis of eigenfunctions.
\begin{corollary} \label{cor:dvolestconvergence}
    Under Assumptions \ref{assump:IndependenceDiffcoordsProjections}, \ref{assump:eigenconvergence} and \ref{assump:PCACoordsConvergence} it follows from Theorem \ref{thm:specvolconvergence}
that
    \begin{eqnarray}
        \dvolest^{(\ell)}
        =
        \Expect[\meig^{(\ell)}( \coords(\vol) + O_{P}(1)) \vol]
        +
        \bigOprob\left(\frac{1}{\sqrt{n}}\right)
    \end{eqnarray}
    where the  expectation is with respect to the distribution of $\vol \in \M$.
    \end{corollary}
So far we have treated the convergence of the spectral volumes.
We now turn to the convergence of the high-resolution reconstructions.
As discussed in Section \ref{sec:SpectralGraphs}, we assume that the manifold of the molecular volumes can be well approximated by a small number of eigenfunctions.  We define this notion precisely in the following assumption.
\begin{assumption} \label{assump:CoefficientDecay}
    There is a set of spectral volumes $\dvol^{(0)}, \ldots, \dvol^{(r-1)}$ and a non-negative function $h(r)$ that satisfies $h(r) \to 0$ such that $h(r)$ bounds the approximation of $\M$ by $r$ spectral volumes:
    \begin{eqnarray} \label{eq:LowdimModel}
        \left\| \vol - \sum_{\ell=0}^{r-1} \dvol^{(\ell)} \meig^{(\ell)}(\coords(\vol)) \right\| = O(h(r)) \qquad \forall \vol \in \M.
    \end{eqnarray}
\end{assumption}
\noindent In that case, we can prove that the true volumes are recovered up to noise.
\begin{theorem} \label{thm:HighResConvergence}
    Consider a sample from a manifold that conforms to \eqref{eq:LowdimModel}, then it follows from Assumptions \ref{assump:IndependenceDiffcoordsProjections}, \ref{assump:eigenconvergence}, \ref{assump:PCACoordsConvergence}  and \ref{assump:CoefficientDecay} that as \( n \to \infty \) we have
\begin{eqnarray}
    \qquad\qquad
    \fl \volest_s = \vol_s + \sum_{\ell=0}^{r-1} O_P(C_\ell)\dvol^{(\ell)} + O_P(h(r)) \left(\sum_{\ell=0}^{r-1} \meig^{(\ell)}(\coords(\vol_s))+ O_P(C_\ell)\right),
\end{eqnarray}
where $C_\ell$ is an upper bound on the norm of the gradient of $\meig^{(\ell)}$.

\end{theorem}
The proof of this theorem is in \ref{sec:proofs}.
Note that the first error term contains an  irreducible error from the finite level of noise in the PCA coordinate assignment.
\section{Results} \label{sec:Results} 

In this section, we apply our method to several synthetic datasets with a low-dimensional conformation space.
We first consider clean images of a clock face with a rotating hand, then more realistic datasets of molecular volumes with one- and two-dimensional motions.

\subsection{Clock dataset} \label{sec:2DClock}
\begin{figure}
    \centering
    \includegraphics[width=3.13cm]{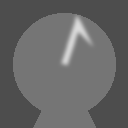}
    \includegraphics[width=3.13cm]{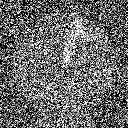}
    \includegraphics[width=3.13cm]{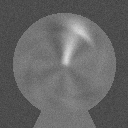}\\ \vspace{1.1mm}
    \includegraphics[width=3.13cm]{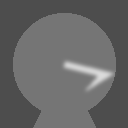}
    \includegraphics[width=3.13cm]{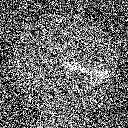}
    \includegraphics[width=3.13cm]{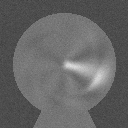}
    \caption{\small (left column) Clean 2D clock faces;
    (middle column) noisy clock faces, used as inputs to the reconstruction algorithm;
    (right column) corresponding reconstructions using $r=15$ spectral volumes.}
    \label{fig:clock2d}
\end{figure}

We begin with a toy model of a 2D clock face with a single moving hand. 
Since the objects we wish to reconstruct are images rather than volumes, no projections are involved.
This is the setting studied in Section \ref{sec:TheoryNoProj}.
The simulated dataset is comprised of $n=10^4$ noisy images ${\bf z}_1, \ldots, {\bf z}_n \in \Real^{N \times N}$ with $N=128$, where each image shows the clock hand at a random angle with additive Gaussian noise.
The affinity matrix was chosen to be
\(
    W_{ij} = \euler^{-\|{\bf z}_i - {\bf z}_j \|^2 / N^2 \sigma^2}
\)
where $\sigma$ is the standard deviation of the Gaussian noise. 
We then constructed a normalized graph Laplacian, extracted its eigenvectors, and computed $r=15$ spectral volumes  by solving the least-squares problem of \eref{eq:DiffvolNoProj}.
Figure \ref{fig:clock2d} shows representative input images and their corresponding reconstructions.
Figure \ref{fig:clock2dspecvols} shows the estimated spectral volumes (spectral images in this case).
We need a large value of $r$ to get good reconstructions, since the clock hand  has sharp discontinuities.

\begin{figure}
    \qquad\quad\ \ 
    \includegraphics[height=3.13cm]{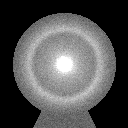}
    \includegraphics[height=3.13cm]{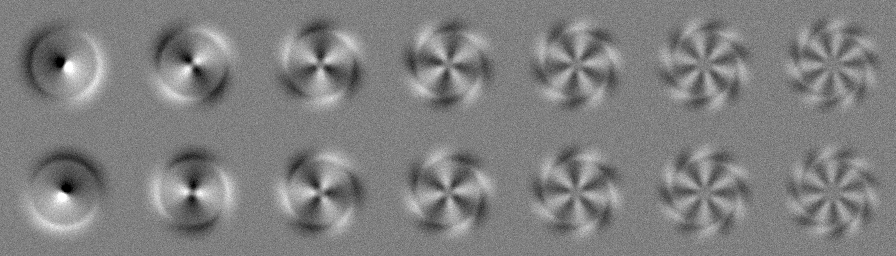}
    \caption{\small Spectral volumes $\dvolest^{(0)}, \ldots, \dvolest^{(15)}$ of the 2D clock. (left) First spectral volume $\dvol^{(0)}$ which converges to the mean; (right) Other spectral
volumes ordered vertically in pairs of the same eigenvalue. Eigenvalues increase from left to right.}
    \label{fig:clock2dspecvols}
\end{figure}

To analyze this example we note that the clock dataset has the manifold geometry of the unit circle $S^1$.
Ignoring an arbitrary phase offset, the set of real eigenfunctions of the Laplace--Beltrami operator on $S^1$ are $\meig^{(0)}(\theta) = 1/2\pi$ and for all integer
$\ell \ge 1$, \( \meig^{(2\ell-1)}(\theta) = \sqrt{1/\pi} \sin(\ell \theta) \) and \( \meig^{(2\ell)}(\theta) = \sqrt{1/\pi} \cos(\ell \theta) \).
It follows from \eref{eq:SpecVolsSolWithoutProj} that
\begin{eqnarray}
    \dvolest^{(\ell)}
    \to
    \Expect[\meig^{(\ell)}({\bf z}) {\bf z}]
    =
    \int_{{\bf z} \sim \M} \meig^{(\ell)}({\bf z}) {\bf z} \mathrm{d}{\bf z}.
\end{eqnarray}
Let ${\bf z}_\theta$ denote the image with the clock hand at angle $\theta$, we may rewrite the above as
\begin{eqnarray}
    \dvolest^{(\ell)}
    \to
    \int_0^{2\pi} \meig^{(\ell)}(\theta) {\bf z}_\theta \mathrm{d}\theta.
\end{eqnarray}
Rather than fixing a pixel and rotating the clock hand, we may fix the clock hand and rotate the pixel in the other direction.
For pixels $[x,y]$ inside the disk of the clock face,
\begin{eqnarray}
    \dvolest^{(\ell)}[x,y]
    \to
    \int_0^{2\pi} \meig^{(\ell)}(\theta) {\bf z}_{\theta=0}[R_{-\theta}[x,y]^T] \mathrm{d}\theta.
\end{eqnarray}
We conclude that in the case of simple rotation heterogeneity, a pixel of the $\ell$\textsuperscript{th} spectral volume in the rotating domain converges to the $\ell$\textsuperscript{th} Fourier coefficient of the function $f(\theta) = {\bf z}_0[R_{-\theta}[x,y]]$.
Put differently, the coefficients $\dvolest^{(0)}[x,y], \dvolest^{(1)}[x,y], \ldots$ converge to the Fourier coefficients of the rotating domain, in polar representation.

We tested a similar clock dataset in 3D, using the same clock hand shape, this time with tomographic projections.
The results we obtained are similar to the results of the 2D clock dataset, in accordance with Corollary \ref{cor:dvolestconvergence}.
See \ref{sec:3DClock} for details.

\subsection{Simulated ion channel} \label{sec:SimIonChannel}
We created two synthetic datasets based on a voltage-gated potassium channel, shown in Figure \ref{fig:FakeKVorig}.
The first dataset \textsf{ChannelSpin} demonstrates a rotational motion of the top part about the $z$ axis.
The second dataset \textsf{ChannelStretch} demonstrates a nonrigid stretching of the bottom part.
Specifically, each slice was displaced along the $x$-$y$
plane by an amount that is proportional to the squared distance of the slice from the center of the molecule.
We used a spatial resolution of $\imsize = 108$ and generated $n = 10,000$ volumes for each dataset.
The angles of rotation in the \textsf{ChannelSpin} dataset were drawn uniformly, which gives a conformational manifold diffeomorphic to the circle.
For the \textsf{ChannelStretch} dataset, we drew random displacements $\delta_x, \delta_y \in \{-16,-15,\ldots,16\}$ that parameterize a non-rigid shift of every $x$-$y$ slice in the bottom half of the molecule.
Let \( {\bf v} \in \mathbb{R}^{\imsize \times \imsize \times \imsize} \) be the original (unstretched) ion channel, the stretched ion channel \( \bf v' \) is defined for every $0 \le z \le \imsize/2$ by   
\begin{eqnarray}
    {\bf v}'[x,y,z] = {\bf v}[x + \delta_x s_z, y + \delta_y s_z, z] \quad \mathrm{where}\quad s_z = \left(\frac{\imsize/2-z}{\imsize/2-z_0}\right)^2
\end{eqnarray}
Note that $z=N/2$ is the center of mass of the ion channel and $z_0=16$ is the bottom of the molecule.

In both datasets, we projected the molecules along random orientations, drawn uniformly from $\SO(3)$.
We then applied a simulated point spread function with a defocus value chosen uniformly at random from $1.50$, $1.67$, $1.83$, $2.00$, $2.17$, $2.33$, or $2.50$ microns.
Finally, we added white Gaussian noise, with a variance chosen such that the total energy of the noise was 30 times that of the total energy
of each clean image.
No  in-plane shift was applied.
See Figure \ref{fig:FakeKVorig} for  example images.

Using the projection images, we ran the covariance estimation method with $q = 4$ components to build the adjacency matrix for the \textsf{ChannelSpin} dataset and $q = 8$ for the \textsf{ChannelStretch} dataset.
We then reconstructed the volumes using $r=1, \ldots, 15$ spectral volumes.
We used the true orientations of the projection images for both the covariance and spectral volume estimation procedures.

\paragraph{Examining the Laplacian eigenmaps embedding}
\begin{figure}
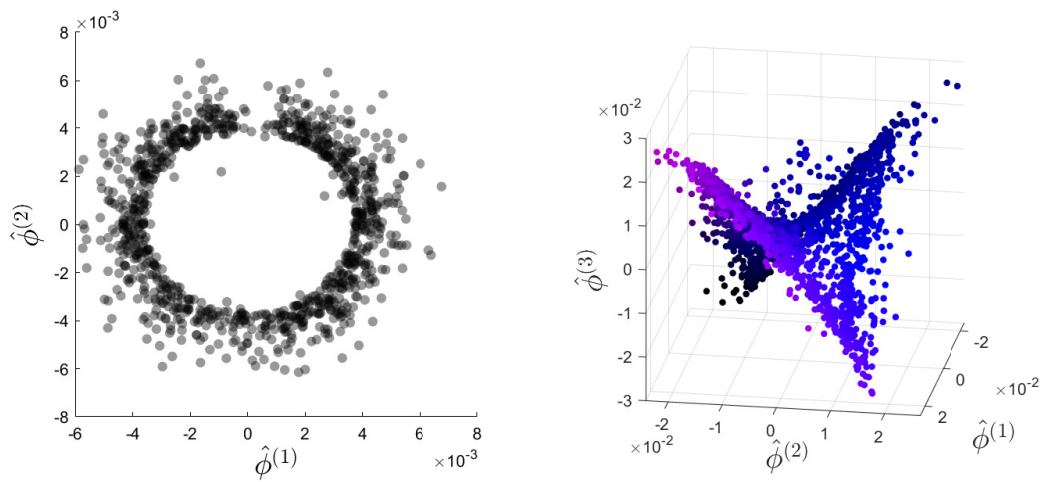

    \begin{center}\qquad\qquad
        \includegraphics[width=69mm]{\detokenize{mani_figure_circle}}\ 
        \includegraphics[width=69mm]{\detokenize{mani_figure_square_123}}
        \caption{\label{fig:ManiEmbeddings} Laplacian eigenmaps embedding of the ion channel datasets. (left) scatter plot of 1000 samples from the \textsf{ChannelSpin} dataset, showing the first two nontrivial Laplacian eigenvector coordinates. (right) 3D scatter plot of 2000 samples from the \textsf{ChannelStretch} dataset, showing the first three nontrivial eigenvectors. This dataset forms a saddle over a 2D square. The blue component of the color is given by the position along the line $\geigs^{(1)} = \geigs^{(2)}$, whereas the red component is given by the position along the line $\geigs^{(1)} = -\geigs^{(2)}$.}
    \end{center}
\end{figure}

Figure \ref{fig:ManiEmbeddings} shows the embeddings of a random sample from the \textsf{ChannelSpin} and \textsf{ChannelStretch} datasets.   The embedding of  \textsf{ChannelSpin} clearly shows a circle whereas the embedding of the \textsf{ChannelStretch} dataset shows a 2-dimensional square in the $ \geigs^{(1)}-\geigs^{(2)}$ plane that is shaped like a saddle.  Both of these results are in accordance with the underlying motion manifold.

\paragraph{Examining the spectral volumes}

Figure \ref{fig:RotSpecVols} shows the first few spectral volumes.
For the \textsf{ChannelSpin} dataset, as expected $\dvolest^{(0)}$ captures the mean over all rotations.
Higher order spectral volumes have increasing angular frequency, capturing more and more detail.
Note that $\dvolest^{(1)}$ and $\dvolest^{(2)}$ are a quarter period out of phase, while $\dvolest^{(3)}$ has twice the angular frequency.
Due to  the $C_4$ symmetry of the ion channel, the lowest frequency of the \textsf{ChannelSpin} dataset has a period of 90 degrees.
\begin{figure}[t]
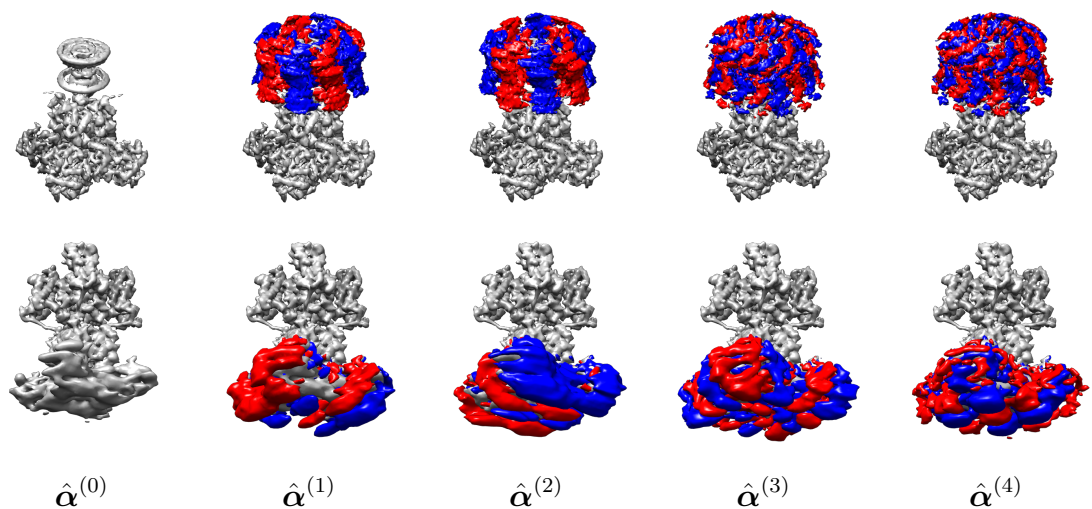

    \qquad
    \begin{tabular}{ccccc}
        \includegraphics[width=26mm]{\detokenize{spinning_alpha_1_pic}}
        &
        \includegraphics[width=26mm]{\detokenize{spinning_alpha_2_pic}}
        &
        \includegraphics[width=26mm]{\detokenize{spinning_alpha_3_pic}}
        &
        \includegraphics[width=26mm]{\detokenize{spinning_alpha_4_pic}}
        &
        \includegraphics[width=26mm]{\detokenize{spinning_alpha_5_pic}}
        \vspace{-0.5em}\\
        \includegraphics[width=26mm]{\detokenize{alpha_0_pic}}
        &
        \includegraphics[width=26mm]{\detokenize{alpha_1_pic}}
        &
        \includegraphics[width=26mm]{\detokenize{alpha_2_pic}}
        &
        \includegraphics[width=26mm]{\detokenize{alpha_3_pic}}
        &
        \includegraphics[width=26mm]{\detokenize{alpha_4_pic}}
        \\
        $\dvolest^{(0)}$
        &
        $\dvolest^{(1)}$
        &
        $\dvolest^{(2)}$
        &
        $\dvolest^{(3)}$
        &
        $\dvolest^{(4)}$
    \end{tabular}
    \caption{\label{fig:RotSpecVols}
    Spectral volumes computed from the \textsf{ChannelSpin} dataset (top) and \textsf{ChannelStretch} dataset (bottom).  The zeroth spectral volume $\dvolest^{(0)}$ (grey) is shown on the  left.  The next four figures, from left to right, show higher order spectral volumes, superimposed  over $\dvolest^{(0)}$.  Red and blue represent negative and positive values of the higher-order spectral volume, respectively.
    The two data sets are viewed from different angles.}
\end{figure}

For the \textsf{ChannelStretch} dataset, we see  that $\dvolest^{(0)}$ captures the fixed part of the molecule with high resolution and shows a ``smeared" bottom portion.
The first and second nontrivial spectral volumes each have a low spatial frequency along the $x$ and $y$ axes.
Higher spectral volumes show higher spatial frequencies.
$\dvolest^{(3)}$ shows a mix of the directions of $\dvolest^{(1)}$ and $\dvolest^{(2)}$ whereas $\dvolest^{(4)}$ is similar to $\dvolest^{(1)}$ but with a double spatial frequency. Recall that by Corollary \ref{cor:dvolestconvergence} up to noise the spectral volumes are $\Expect[\meig^{(\ell)}( \coords(\vol)) \vol]$. In this case the eigenfunctions of the Laplacian on the square are the 2D discrete cosine transform basis functions, which are up to scale $ \meig^{(n_x,n_y)} = \cos(n_x x)\cos(n_y y)$ with eigenvalue $\propto n_x^2 + n_y^2$, See \cite[Section 3.1]{GrebenkovNguyen2013}.
This  agrees with our empirical observations.

\paragraph{Reconstruction accuracy}
\begin{figure}[t]
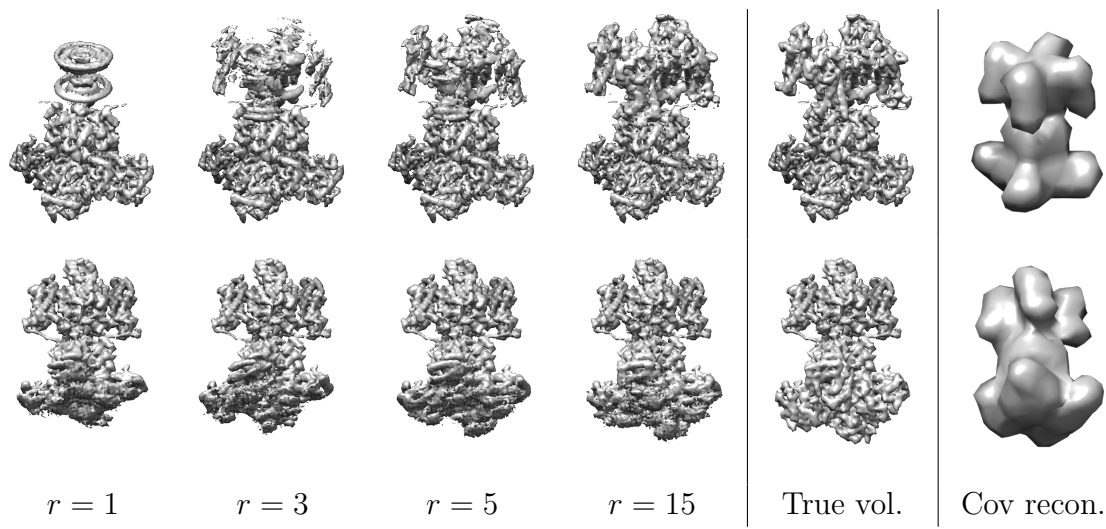

        \qquad
        \begin{tabular}{cccc|c|c}
            \includegraphics[width=21mm]{\detokenize{recon_25000_r_1_crop}}
            &
            \includegraphics[width=21mm]{\detokenize{recon_25000_r_3_crop}}
            &
            \includegraphics[width=21mm]{\detokenize{recon_25000_r_5_crop}}
            &
            \includegraphics[width=21mm]{\detokenize{recon_25000_r_15_crop}}
            &
            \includegraphics[width=21mm]{\detokenize{orig_27_crop}}
            &
            \includegraphics[width=21mm]{\detokenize{recon_cov_crop}}
            \\
            \includegraphics[width=21mm]{\detokenize{square_recon_0}}
            &
            \includegraphics[width=21mm]{\detokenize{square_recon_2}}
            &
            \includegraphics[width=21mm]{\detokenize{square_recon_4}}
            &
            \includegraphics[width=21mm]{\detokenize{square_recon_14}}
            &
            \includegraphics[width=21mm]{\detokenize{square_orig_679}}
            &
            \includegraphics[width=21mm]{\detokenize{cov_recon_679}}
            \\
            $r=1$ & $r=3$ & $r=5$ & $r=15$ & True vol. & Cov recon.
        \end{tabular}
        \caption{\small  \label{fig:OrigAndRecons}
                Reconstructed volumes from \textsf{ChannelSpin} (top row) and \textsf{ChannelStretch} (bottom row), using $r \in \{1, 3, 5, 15\}$ spectral volumes.
 Also shown are the low-resolution reconstructions of the covariance-based method described in Section \ref{sec:LowResHeterogeneity}.
        }
\end{figure}

Figure \ref{fig:OrigAndRecons} shows reconstructions with increasing numbers of spectral volumes alongside the original simulated volume.
Note that the reconstructions for the \textsf{ChannelSpin} volume are of higher quality than for the \textsf{ChannelStretch} volume.
This is expected, since the manifold of conformations of \textsf{ChannelSpin} is one-dimensional whereas \textsf{ChannelStretch} has two-dimensional motion. Hence more samples are needed to get a dense cover of the conformational manifold.

To quantify the accuracy of our reconstructions, we use the \textit{Fourier shell correlation} (FSC), which is the standard evaluation criterion in the cryo-EM literature \cite{HendersonEtal2012}.
Given two volumes $\vol_1, \vol_2$, the FSC takes their Fourier transforms and computes the correlation between each frequency shell.
Because we wish to estimate the quality of reconstructing the variable part of the molecule, instead of reporting the FSC between the reconstructed volumes and the original volumes, we report the mean-subtracted FSC
\begin{eqnarray}
    \mathrm{FSC}\left(\vol_s - \mean, \sum_{\ell = 1}^{r-1} \geigs^{(\ell)} \dvolest^{(\ell)}\right),
\end{eqnarray}
Figure \ref{fig:FSCs} shows the results for each simulated dataset.  As expected, the  reconstruction quality increases with the number of spectral volumes.
The reconstruction of the high frequencies is less accurate than that of the low frequencies.
In both cases, as $r$ increases, the FSC curves converge, suggesting a number after which more spectral volumes yield diminishing returns.
\begin{figure}
    \qquad\quad\ 
    \includegraphics[width=64mm]{\detokenize{hollow_fscs_snr_30}}\qquad
    \includegraphics[width=64mm]{\detokenize{hollow_fscs_snr_30_square}}
    \caption{\label{fig:FSCs}
                FSC curves for $r=2,\ldots,16$ from bottom to top, comparing reconstructed volumes to originals.
                Each curve is the FSC of $\vol_s-\mean$ and $\sum_{\ell=1}^{r}  \geigs^{(\ell)} \dvolest^{(\ell)}$, averaged over $s=1,\ldots,n$.
                (left) \textsf{ChannelSpin}; (right) \textsf{ChannelStretch}.  The bottom $x$ axis denotes spatial frequency, and the top $x$ axis the corresponding wavelength, so that the rightmost position is the Nyquist frequency and a wavelength of 2 pixels.}
\end{figure}

\subsection{Runtime}\label{sec:Runtime}
Table 2 details the running time of our method for the \textsf{ChannelSpin} simulation with $N=108$, $r=15$, and $q=8$. The method is implemented in MATLAB 2017b and runs on 16 cores of a 2.3 GHz Intel Xeon CPU; memory usage was about 60 GB.

\begin{table}\small 
\caption{\small \label{table:RunTimes}Runtimes for the main steps of our method on the \textsf{ChannelSpin} dataset, with $n=10,000$ images of  $108 \times 108$ pixels.}
\begin{center}
\begin{tabular}{l|r}
    Procedure & Running time (sec) \\
    \hline
    Calculation of $\meanest$ & 624.6 \\
    Calculation of $\covest$ & 5044.7 \\
    Calculation $\eigvest$ & 0.8 \\
    Calculation of \{$\estcoords_s\}$ & 2084.8 \\
    Calculation of $\{\geig_s\}$ & 531.5 \\
    Calculation of $\wtskerest$ & 12378.0 \\
    Calculation of $\wtsrhsest$ & 4014.1 \\
    Estimation of $\{\dvolest^{(\ell)}\}_{\ell=0}^{15}$ & 1769.9 \\
    \hline
\end{tabular}
\end{center}
\end{table}

\section{Conclusion} \label{sec:Conclusion}
Today, rigid macromolecules are routinely reconstructed to near-atomic resolution using standard cryo-EM software tools.
However, the high-resolution reconstruction of molecular samples with continuous heterogeneity remains one of the grand challenges of the field.
This work describes a new method which addresses this challenge.
It combines spectral graph theory  with  recent techniques for covariance-based low-resolution reconstruction.
Our procedure computes conformation-dependent Laplacian eigenmap coordinates and then generates a set of spectral volumes that characterize the variability of the molecule under study.
Together these define a high-resolution 3D reconstruction for every projection image. 

In the context of machine learning, our method combines and extends two classical methods: (i) the low-resolution covariance-based reconstruction, which we use to form the affinity graph, may be viewed as a generalization of PCA, as it finds the principal volumes in the space of molecular conformations. Unlike PCA, the input is projection images rather than full observations. (ii) the construction of an affinity graph from the low-resolution reconstructions and the generalized tomographic reconstruction using the eigenvectors of the graph Laplacian.
This can be viewed as an extension of standard approaches for nonparametric regression, semi-supervised learning and matrix completion on graphs and manifolds (see for example \cite{BelkinNiyogi2004,ZhouSrebro2011,LeeIzbicki2016,MoscovichJaffeNadler2017,VilloutreixEtal2017}).
The key difference is that rather than partially-labeled or noisy observations,  we reconstruct a smooth high-dimensional function from noisy tomographic measurement.
We note that the combination of PCA and graph Laplacian representations has been used for dimensionality
reduction and denoising, for example in \cite{SingerWu2013,Singer2006a}.

Similar to the hyper-molecules method proposed by \cite{LedermanSinger2017,LedermanAndenSinger2019}, our method expands the molecular volumes which generated the projection images using a small set of basis volumes.
However, in the hyper-molecules model the basis volumes are obtained from a user-specified manifold.  Similarly, the multi-body refinement of RELION 3 \cite{NakaneEtal2018} requires that the user manually segment the molecule into components that exhibit motion relative to each other.  In contrast, our method is data-driven and requires no such user input.
It relies only on the assumption that the molecule deforms continuously in a manner that is  determined by a small number of parameters.

To conclude, in this paper we have described a method for the reconstruction of molecules with continuous heterogeneity, studied it theoretically and demonstrated the high-resolution reconstruction of synthetic data with one-dimensional and two-dimensional motion manifolds.
In future work, we will continue to scale up the method and apply it to the analysis of experimental datasets.

\section*{Software}
Code for computing spectral volume reconstructions and  producing the figures in this paper is available at \url{http://github.com/PrincetonUniversity/specvols}

\ack
We would like to thank Ronald R. Coifman and Mark Tygert for interesting discussions.
AS, AM, and AH were partially supported by NIGMS Award Number R01GM090200, AFOSR FA9550-17-1-0291, ARO W911NF-17-1-0512, Simons Investigator Award, the Moore Foundation Data-Driven Discovery Investigator Award, and NSF BIGDATA Award IIS-1837992.
The Flatiron Institute is a division of the Simons Foundation.
3D molecular graphics were rendered using  UCSF Chimera \cite{Chimera2004}.

\section*{References}

\bibliographystyle{unsrtwithdoi}
\bibliography{cryohet}

\appendix

\section[Appendix: Proofs]{Proofs} \label{sec:proofs}

Before proving Theorem \ref{thm:specvolconvergence}, we require a technical lemma that bounds the variance of linear combinations of random variables.
\begin{lemma} \label{lemma:VarianceConvexCombination}
    Fix $n$ and let $Z_1, \ldots, Z_n$ be i.i.d. random variables with finite variance.
    Let $W_1 \ldots, W_n$ be identically distributed, but possibly dependent, random weights.
    Denote the (unnormalized) sample moments by
    \begin{eqnarray}
        \sum_{s=1}^n W_s = m_1
        \qquad
        \sum_{s=1}^n W_s^2 = m_2.
    \end{eqnarray}
    If the weights $W_1, \ldots, W_n$ are independent of $Z_1, \ldots, Z_n$, we have
    \begin{eqnarray}
        \Var
        \left(
            \sum_{s=1}^n
            W_s Z_s
        \right)
        \le
        m_2 \Expect[Z_1^2].
    \end{eqnarray}
\end{lemma}
\begin{proof}
    \ By definition,
    \begin{eqnarray} \label{eq:VarSumWiZi}
        \Var
        \left(
            \sum_{s=1}^n
            W_s Z_s
        \right)
        &=&
        \Expect
        \left[
            \left(
                \sum_{s=1}^n W_s Z_s
            \right)^2
        \right]
        -
        \left(
        \Expect
        \left[
            \sum_{s=1}^n W_s Z_s
        \right]
        \right)^2.
    \end{eqnarray}
    Since $W_s$ and $Z_s$ are independent, we can rewrite the second term, which yields
    \begin{eqnarray}
        \left(
            \Expect
            \left[
                \sum_{s=1}^n W_s Z_s
            \right]
        \right)^2
        =
        \left(
            \sum_{s=1}^n
            \Expect[W_s]
            \Expect[Z_s]
        \right)^2
        =
        m_1^2
        \left(
            \Expect [Z_1]
        \right)^2.
    \end{eqnarray}
    Similarly, $W_s W_t$ is independent of $Z_1, \ldots, Z_n$, so we may break up the expectations in the first term of \eref{eq:VarSumWiZi}.
    We then split the double sum into a sum over index pairs $s=t$ and a sum over $s \neq t$, obtaining
    \begin{eqnarray}
        \Expect
        \left[
            \left(
                \sum_{s=1}^n W_s Z_s
            \right)^2
        \right]
        &=&
        \sum_{s,t=1}^n
        \Expect[W_s W_t] \Expect[Z_s Z_t]
        \\
        &=&
        \sum_{s=1}^n \Expect[W_s^2] \Expect[Z_s^2]
        +
        \sum_{t=1}^n \sum_{s \neq t} \Expect[W_s W_t] \Expect[Z_s Z_t]
        \\
        &=&
        \Expect[Z_1^2] \sum_{s=1}^n \Expect[W_s^2]
        +
        \left( \Expect[Z_1] \right)^2 \sum_{t=1}^n \sum_{s \neq t} \Expect[W_s W_t].
    \end{eqnarray}
    The second term may be bounded by the constraint $\sum_{s=1}^n W_s = m_1$, since
    \begin{eqnarray}
         \hspace{-10mm}
         \left(\Expect[Z_1] \right)^2
         \Expect \left[ \sum_{t=1}^n \sum_{s \neq t} W_s W_t \right]
         \le
         \left(\Expect[Z_1] \right)^2
         \Expect \left[ \left( \sum_{s=1}^n W_s \right)^2 \right]
         =
         \left(\Expect[Z_1] \right)^2 m_1^2.
    \end{eqnarray}
    Putting it all together and incorporating the second moment constraint, we obtain
    \begin{eqnarray}
        \Var
        \left(
            \sum_{s=1}^n
            W_s Z_s
        \right)
        \le
        \Expect[Z_s^2]\sum_{s=1}^n \Expect[W_s^2]
        \le
        m_2 \Expect[Z_1^2]. 
    \end{eqnarray}
\end{proof}
\noindent \textbf{Proof of Theorem \ref{thm:specvolconvergence}.}

\noindent We will prove the convergence of the solution to \eref{eq:DiffvolEquation} by proving that both $\wtsrhsest$ and $\wtskerest$ converge in probability as $n \to \infty$.
We start by computing the expectation of $\wtsrhsest$.
\begin{eqnarray} 
    \Expect[\wtsrhsest^{(\ell)}]
    &=&
    \Expect
    \left[
        \frac{1}{\sqrt{n}}
        \sum_{s=1}^n
        \geigs^{(\ell)} \proj_s^\transp \im_s^{}
    \right]\\
    &=&
    \sum_{s=1}^n \frac{1}{\sqrt{n}} \Expect[ \geigs^{(\ell)} \proj_s^\transp \proj_s^{} \vol_s^{}  + \geigs^{(\ell)} \proj_s^\transp \noise_s^{}] \qquad  \mathrm{(By\ \eqref{eq:ForwardModel})} \\
    &=&
    \sum_{s=1}^n \frac{1}{\sqrt{n}} \Expect[ \geigs^{(\ell)} \proj_s^\transp \proj_s^{} \vol_s^{} ] + 0 \\
    &=&
    \Expect[\proj^\transp \proj^{}]
    \Expect
    \left[
        \frac{1}{\sqrt{n}}
        \sum_{s=1}^n
        \geigs^{(\ell)} \vol_s^{}
    \right] \qquad \mathrm{(By\ Assumption\ \ref{assump:IndependenceDiffcoordsProjections})}.
\end{eqnarray}
Now consider the variance  of the $i$\textsuperscript{th} element of the vector \( \wtsrhsest^{(\ell)} \),
\begin{eqnarray} \label{eq:VarWeightedBackproj}
    \Var({\bf e}_i^\transp \wtsrhsest^{(\ell)})
    =
    \Var
    \left(
        \frac{1}{\sqrt{n}}
        \sum_{s=1}^n \geigs^{(\ell)} {\bf e}_i^\transp \proj_s^\transp \im_s^{}
    \right).
\end{eqnarray}
We apply Lemma \ref{lemma:VarianceConvexCombination} with $W_s = \geigs^{(\ell)}$ and $Z_s = {\bf e}_i^\transp \proj_s^\transp \im_s^{}$ to obtain
\begin{eqnarray}
    \Var({\bf e}_i^\transp \wtsrhsest^{(\ell)}) <  \frac{1}{n} \Expect[Z_1^2].
\end{eqnarray}
We now compute the expectation and variance of the matrix \( \wtskerest \).
\begin{eqnarray}
    \Expect
    \left[
        \wtskerest^{(\ell, m)}
    \right]
    &=&
    \sum_{s=1}^n
    \Expect
    \left[
    \geigs^{(\ell)}
    \geigs^{(m)}
    \proj_s^\transp
    \proj_s^{}
    \right]
    \\
    &=&
    \sum_{s=1}^n
    \Expect[\geigs^{(\ell)} \geigs^{(m)}]
    \Expect[\proj_s^\transp \proj_s^{}]
    \qquad \mathrm{(By\ Assumption\ \ref{assump:IndependenceDiffcoordsProjections})}
    \\
    &=&
    \Expect
    \left[
        \sum_s
        \geigs^{(\ell)} \geigs^{(m)}
    \right]
    \Expect[\proj^\transp \proj^{}]
    =
    \delta_{\ell,m} \Expect[\proj^\transp \proj^{}].
\end{eqnarray}
For the variance, we compute the variance of a single entry \( \wtskerest^{(\ell, m)}_{i,j} = \mathbf{e}_i^\transp \wtskerest^{(\ell, m)}
\mathbf{e}_j\).

\noindent \textbf{Case 1: $\ell=m$}
\begin{eqnarray}
    \Var
    \left(
        \mathbf{e}_i^\transp \wtskerest^{(\ell, \ell)} \mathbf{e}_j
    \right)
    &=&
    \Var
    \left(
        \sum_{s=1}^n
        \left( \geigs^{(\ell)} \right)^2
        \mathbf{e}_i^\transp
        \proj_s^\transp
        \proj_s^{}
        \mathbf{e}_j^{}
    \right).
\end{eqnarray}
Let $W_s = \left( \geigs^{(\ell)} \right)^2$ and $Z_s = {\bf e}_i^\transp \proj_s^\transp \proj_s^{} {\bf e}_j$.
By Assumption \ref{assump:EigenvectorBound}, \( \sum_{s=1}^n W_s^2 \to 0\).
Since $Z_s$ has finite variance, we apply Lemma \ref{lemma:VarianceConvexCombination} to obtain that as $n$ tends to infinity, \( \Var(\mathbf{e}_i^\transp \wtskerest^{(\ell, \ell)} \mathbf{e}_j) \to 0 \) in probability.

\noindent \textbf{Case 2: $\ell \neq m$}
\begin{eqnarray}
    \Var
    \left(
        \mathbf{e}_i^\transp \wtskerest^{(\ell, m)} \mathbf{e}_j
    \right)
    &=&
    \Var
    \left(
        \sum_{s=1}^n
        \geigs^{(\ell)}
        \geigs^{(m)}
        \mathbf{e}_i^\transp
        \proj_s^\transp
        \proj_s^{}
        \mathbf{e}_j^{}
    \right).
\end{eqnarray}
By Cauchy--Schwarz and Assumption \ref{assump:EigenvectorBound}, the following converges in probability:
\begin{eqnarray}
    \sum_{s=1}^n
    \left(
        \geigs^{(\ell)}
    \right)^2
    \left(
        \geigs^{(m)}
    \right)^2
    \le
    \sqrt{\sum_{s=1}^n \left( \geigs^{(\ell)} \right)^4}
    \sqrt{\sum_{s=1}^n \left( \geigs^{(m)} \right)^4}
    \to
    0.
\end{eqnarray}   
Again, we apply Lemma \ref{lemma:VarianceConvexCombination}. This time with $W_s = \geigs^{(\ell)}
\geigs^{(m)}$ and $Z_s = {\bf e}_i^\transp
\proj_s^\transp \proj_s^{} {\bf e}_j^{}$ to obtain that \( \Var(\mathbf{e}_i^\transp \wtskerest^{(\ell, m)} \mathbf{e}_j) \to 0 \) in probability.
To summarize, we  proved the following results:
\begin{eqnarray} \label{eq:bEll}
    \wtsrhsest^{(\ell)}
    &=&
    \Expect[\proj^\transp \proj^{}]
    \Expect
    \left[
        \frac{1}{\sqrt{n}}
        \sum_{s=1}^n \geigs^{(\ell)} \vol_s^{}
    \right]
    +
    \bigOprob \left( \frac{1}{\sqrt{n}} \right)
    \\
    \wtskerest^{(\ell, m)}
    &\to&
    \delta_{\ell,m}\Expect[\proj^\transp \proj^{}] \quad \mathrm{In\ probability}. \label{eq:Klm}
\end{eqnarray}
By \eref{eq:DiffvolEquation}, the vector of spectral volumes satisfies
\(
    \dvol = \wtskerest^{-1} \wtsrhsest.
\)
Denote \( \wtskerest = \Expect[\wtskerest] + \Delta \wtskerest \) and \( \wtsrhsest = \Expect[\wtsrhsest] + \Delta \wtsrhsest \).
Expanding \( \wtskerest^{-1} \) to first order,
\begin{eqnarray}
    \wtskerest^{-1} = \Expect[\wtskerest]^{-1} + \Expect[\wtskerest]^{-1} \, \Delta \wtskerest \, \Expect[\wtskerest]^{-1} + \bigO(\| \Delta \wtskerest\|^2).
\end{eqnarray}
Since $\Delta \wtskerest \to 0$ and $\Delta {\bf b} = O_P(1/\sqrt{n})$, the spectral volumes satisfy,
\begin{eqnarray}
    \dvol = (\Expect[\wtskerest] + \Delta \wtskerest)^{-1} (\Expect[\wtsrhsest] + \Delta \wtsrhsest) =  \Expect[\wtskerest]^{-1} \Expect[\wtsrhsest] + \bigOprob(1/\sqrt{n}).
\end{eqnarray}
Plugging in Equations \eqref{eq:bEll} and \eqref{eq:Klm}, we obtain
\begin{eqnarray}
    \dvolest^{(\ell)}
    =
    \Expect
    \left[
        \frac{1}{\sqrt{n}}
        \sum_{s=1}^n
        \geigs^{(\ell)} \vol_s
    \right]
    +
    \bigOprob \left( \frac{1}{\sqrt{n}} \right) .
\end{eqnarray}
\( \hfill \square \)

\noindent \textbf{Proof of Theorem \ref{thm:HighResConvergence}.}
It follows from Corollary \ref{cor:dvolestconvergence} and Assumption \ref{assump:CoefficientDecay} that
\begin{eqnarray}
    \fl \quad\,
    \dvolest^{(\ell)}
    =
    \Expect[\meig^{(\ell)}( \coords(\vol) + O_P(1)) \vol] + O_P(1/\sqrt{n})
    \\
    \fl \qquad\quad=
    \Expect \left[ \meig^{(\ell)}(\coords(\vol) + O_P(1)) \left( O(h(r)) + \sum_{m=0}^{r-1} \dvol^{(m)} \meig^{(m)}(\coords(\vol))  \right) \right] + O_P \left(\frac{1}{\sqrt{n}}\right)     \\
\label{eq:DvolestConvDvol}
\end{eqnarray}
Note that $\meig^{(\ell)}$ is a smooth function on a compact domain and therefore
its derivatives are
bounded. Hence \( \meig^{(\ell)}(\coords(\vol) + O_P(1)) = \meig^{(\ell)}(\coords(\vol)) + O_P(1) \).
It follows that,
\begin{eqnarray}
    \fl \quad\,
    \dvolest^{(\ell)}
    =
    \Expect \left[ O_P(h(r)) + \sum_{m=0}^{r-1} \dvol^{(m)} \meig^{(\ell)}(\coords(\vol))\meig^{(m)}(\coords(\vol))
  \right] + O_P \left(\frac{1}{\sqrt{n}}\right)
    \\
    \fl \qquad \quad=
    \dvol^{(\ell)} + O_P(h(r)). \qquad \mathrm{(by\ } \Expect[ \meig^{(\ell)} \meig^{(m)} ] = \delta_{\ell,m} \mathrm{)}
\end{eqnarray}
By the definition of the high-resolution reconstructions \eref{eq:highresvolest} we now have
\begin{eqnarray}
    \volest_s
    =
    \sqrt{n}\sum_{\ell=0}^{r-1} \geigs^{(\ell)} \dvolest^{(\ell)}
    =
    \sqrt{n}\sum_{\ell} \geigs^{(\ell)} (\dvol^{(\ell)} + O_P(h(r))).
\end{eqnarray}
By Assumption \ref{assump:eigenconvergence} we may rewrite this as
\begin{eqnarray} \label{eq:HighQReconApprox}
    \volest_s
    =
    \sum_{\ell} (\meig^{(\ell)}(\estcoords_s) + o_P(1))(\dvol^{(\ell)} + O_P(h(r))).
\end{eqnarray}
By Assumption \ref{assump:PCACoordsConvergence} we have
\( 
    \meig^{(\ell)}(\estcoords_s)
    =
    \meig^{(\ell)}(\coords(\vol_s) + O_P(1))
    =
    \meig^{(\ell)}(\coords(\vol_s)) + O_P(C_\ell),
\)
where the last equality stems from the fact that $\meig^{(\ell)}$ is a smooth function on a compact domain and therefore its derivatives are
bounded.
We plug this back into \eref{eq:HighQReconApprox},
\begin{eqnarray}
    \volest_s
    &=
    \sum_{\ell} (\meig^{(\ell)}(\coords(\vol_s)) + O_P(C_\ell))(\dvol^{(\ell)} + O_P(h(r)))
    \\
    &=
    \sum_{\ell} \meig^{(\ell)}(\coords(\vol_s))\dvol^{(\ell)} +  \sum_{\ell} O_P(C_\ell)\dvol^{(\ell)}
    \\
    &+ O_P(h(r)) \sum_{\ell} \meig^{(\ell)}(\coords(\vol_s))
    + \sum_\ell O_P(C_\ell) O_P(h(r)).
\end{eqnarray}
We conclude the proof by reusing Assumption \ref{assump:CoefficientDecay} on the first term.
\( \hfill \square \)

\section[Appendix: 3D Clock]{3D Clock} \label{sec:3DClock}

\begin{figure}[ht]
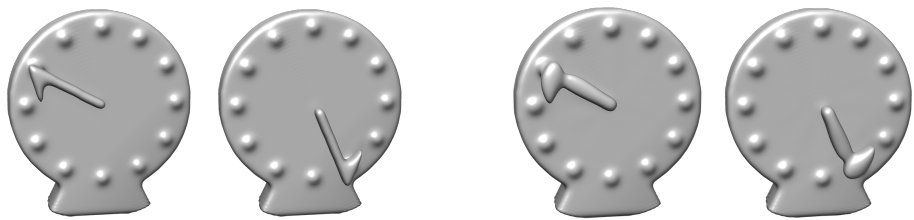

        \qquad\qquad\qquad
        \includegraphics[height=32mm]{\detokenize{orig_clock_1_crop}}
        \includegraphics[height=32mm]{\detokenize{orig_clock_25_crop}}\qquad\quad
        \includegraphics[height=32mm]{\detokenize{recon_clock_1_crop}}
        \includegraphics[height=32mm]{\detokenize{recon_clock_10000_crop}}
        \caption{\small \label{fig:3DClockRecon}
                (left) Two conformations of the clock model; (right) Their reconstructions using $r = 7$ spectral volumes.
        }
\end{figure}
\begin{figure}[ht]
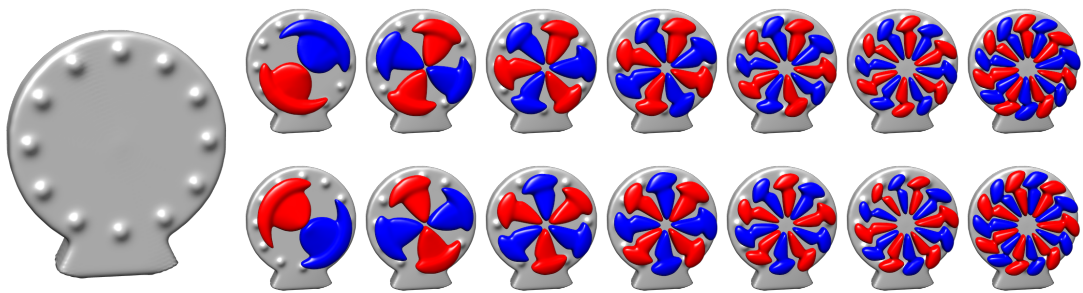

\qquad\quad
\begin{minipage}{0.20\linewidth}
    \includegraphics[width=32mm]{\detokenize{spec_vol_0}}
\end{minipage}
\begin{minipage}{0.78\linewidth}
    \setlength\tabcolsep{0pt}
        \begin{tabular}{ccccccc}
                \includegraphics[width=16mm]{\detokenize{spec_vol_1}}
                &
                \includegraphics[width=16mm]{\detokenize{spec_vol_3}}
                &
                \includegraphics[width=16mm]{\detokenize{spec_vol_5}}
                &
                \includegraphics[width=16mm]{\detokenize{spec_vol_7}}
                &
                \includegraphics[width=16mm]{\detokenize{spec_vol_9}}
                &
                \includegraphics[width=16mm]{\detokenize{spec_vol_11}}
                &
                \includegraphics[width=16mm]{\detokenize{spec_vol_13}}
                \\
                \includegraphics[width=16mm]{\detokenize{spec_vol_2}}
                &
                \includegraphics[width=16mm]{\detokenize{spec_vol_4}}
                &
                \includegraphics[width=16mm]{\detokenize{spec_vol_6}}
                &
                \includegraphics[width=16mm]{\detokenize{spec_vol_8}}
                &
                \includegraphics[width=16mm]{\detokenize{spec_vol_10}}
                &
                \includegraphics[width=16mm]{\detokenize{spec_vol_12}}
                &
                \includegraphics[width=16mm]{\detokenize{spec_vol_14}}

        \end{tabular}
\end{minipage}
        \caption{\small  \label{fig:3DClockSpecVols}
                Surface plots of the 3D clock spectral volumes, shown superimposed on $\dvol^{(0)}$ to aid in context.  Grey is $\dvol^{(0)}$, red
and blue are correspond to negative and positive portions of the higher order spectral volume, respectively.}
\end{figure}

\noindent This appendix presents the 3D analog to the 2D clock simulation.
Here the spatial resolution is $\imsize = 256$ and the number of volumes is $n = 10^5$.
Clock hand angles were drawn uniformly at random from the circle and the viewing orientations were drawn at uniformly from $\SO(3)$.
No noise was added and no CTF was applied in order to test the behavior under ideal conditions.
The covariance estimation method was run with $q = 8$ components in order to generate the adjacency matrix.
We then constructed a symmetric normalized graph Laplacian and performed reconstructions with $r = 15$ spectral volumes.
Figure \ref{fig:3DClockSpecVols} shows the spectral volumes for the 3D clock. As can be seen, the spectral volumes here resemble very closely
those from Section \ref{fig:clock2dspecvols}, limited to the region that the clock hand rotates in.
The zeroth spectral volume looks like the mean volume whereas higher order spectral volumes come in pairs of increasing angular
frequency.

\end{document}